\documentclass[11pt]{article}
\usepackage[T1]{fontenc}
\usepackage{graphicx}
\usepackage{natbib}
\usepackage{authblk}
\usepackage[a4paper,margin=1in]{geometry}
\usepackage{amsmath,amstext,amsthm,amssymb,amsfonts,physics,bbm}
\usepackage{enumitem}
\usepackage{comment}
\usepackage{framed}
\usepackage{algorithm}
\usepackage[noend]{algpseudocode}
\usepackage{xcolor}
\usepackage{tikz}

\usepackage{framed}
\usepackage{caption}

\definecolor{xdxdff}{rgb}{0.5,0.5,1}
\usepackage{pgfplots}
\pgfplotsset{compat=1.15}
\usepackage{mathrsfs}
\usepackage[margin=1in]{geometry}
\usetikzlibrary{arrows}
\usetikzlibrary{shapes}
\usetikzlibrary{plotmarks}
\usetikzlibrary{arrows}
\usetikzlibrary{trees}
\usepackage{tikz-qtree}
\usepackage{float}
\usepackage{pgfplots}
\usepackage{mathrsfs}
\usepackage{graphicx}
\usepackage{fullpage}
\newtheorem{theorem}{Theorem}
\usepackage{pgfplots} \pgfplotsset{compat=1.15} \usepackage{mathrsfs}
\usepackage{tabularx}
\usetikzlibrary{arrows}
\usepackage{thmtools,thm-restate} 
\usepackage{caption}
\PassOptionsToPackage{pagebackref}{hyperref} 
\RequirePackage{hyperref}
\hypersetup{
    colorlinks=true,
    linkcolor=red, 
    citecolor=blue, 
    filecolor=cyan, 
    urlcolor=ForestGreen, 
}
\renewcommand*{\backref}[1]{}
\renewcommand*{\backrefalt}[4]{
    \ifcase #1 (Not cited.)
    \or        (#2)
    \else      (#2)
    \fi
}
\RequirePackage[capitalize, noabbrev]{cleveref}

\newtheorem{clm}{Claim}
\newtheorem{lemma}{Lemma}
\newtheorem{definition}{Definition}

\newtheorem{prop}{Proposition}
\newtheorem{example}{Example}[section]

\newtheorem{property}{Property}

\newcommand{\points}{\mathcal{P}}
\newcommand{\strat}{\mathcal{S}}

\newcommand{\mcV}{\mathcal{V}}
\newcommand{\mcI}{\mathcal{I}}
\newcommand{\integers}{\mathbb{Z}}

\newcommand{\rationals}{\mathbb{Q}}
\newcommand{\Description}[1]{}

\DeclareMathOperator{\poly}{poly}
\DeclareMathOperator{\near}{near}
\DeclareMathOperator{\util}{u}
\DeclareMathOperator{\utilout}{u^{\text{new}}}
\DeclareMathOperator{\utileft}{u^{\text{L}}}
\DeclareMathOperator{\utilright}{u^{\text{R}}}
\DeclareMathOperator{\mone}{\mathbbm{1}}

\DeclareMathOperator{\intr}{best-r}
\DeclareMathOperator{\intu}{best-u}
\DeclareMathOperator{\minu}{min-u}
\DeclareMathOperator{\maxd}{max-d}

\DeclareMathOperator{\bit}{bit}
\DeclareMathOperator{\mean}{\mu}
\DeclareMathOperator{\ind}{ind}
\begin{document}
\title{Equilibrium Computation in the Hotelling-Downs Model \\ of Spatial Competition\footnote{Accepted to appear in ACM EC 2026.}}

\author{Umang Bhaskar\thanks{umang@tifr.res.in} }
\author{Soumyajit Pyne\thanks{soumyajit.pyne@tifr.res.in}}
\affil{Tata Institute of Fundamental Research, Mumbai}
\maketitle              

\begin{abstract}

    The Hotelling-Downs model of spatial competition is a foundational framework for understanding strategic positioning by candidates in political economy and industrial organization. In this model, voters are distributed on a line, representing their ideological position on an issue. Each candidate then chooses a position on the line to maximize her vote share. Voters vote for the nearest candidate, closest to their ideological position. This sets up a game between the candidates. While the existence and qualitative structure of equilibria in this model have been extensively studied, the computational complexity remains poorly understood.

We develop both structural results and algorithms for equilibrium computation in this model, giving the first algorithmic results for equilibrium computation in both discrete and continuous variants of the model. Structurally, for the continuous variant, we show that if there exists an equilibrium, there is also one where all candidate positions are rational numbers of bounded complexity, allowing us to discretize the search space. To obtain the structural result we introduce the notion of reflection barriers, which bound local deviations by candidates at an equilibrium.

We then use a dynamic programming approach, to obtain the following algorithmic results. First, when the set of possible candidate locations is discrete, equilibrium existence can be decided in polynomial time. Second, when candidate locations are continuous and voter locations are discrete, we use the previous structural result to give a pseudopolynomial time algorithm. Third, in the fully continuous setting where all equilibria could require irrational numbers, we show that an $\epsilon$-equilibrium can be computed in time polynomial in $1/\epsilon$.
\end{abstract}

\section{Introduction}

The Hotelling-Downs model of spatial competition, along with its variants and extensions, is possibly the most widely studied model of strategic positioning by candidates in social choice theory~\cite{Hotelling29,Downs57}. The basic model considers voters to be distributed arbitrarily in a bounded interval (say $[0,R]$) on a line. A set of $m$ candidates then choose their positions in $[0,R]$, with each candidate trying to maximize her vote share, with the assumption that each voter votes for the nearest candidate. While Hotelling introduced the model to study pricing and location of firms in a market, Downs noted its application to political competition as well.

The model thus sets up a game between the candidates, each choosing as a strategy her location in $[0,R]$ to maximize her vote share. A pure Nash equilibrium (PNE) is the natural outcome studied for this model. Initially, the model was studied for just two candidates. In this case, both candidates position themselves at the median of the set of voters. For three candidates, an equilibrium may not exist~\cite{EatonL75,Osborne}, and necessary and sufficient conditions for the existence of equilibria are known~\cite{EatonL75,hayashi2019necessary}.

Despite the nonexistence of equilibrium, the model is a fundamental and intuitive tool for understanding and analysing strategic positioning by candidates and firms. For example, the strategic behavior of candidates and its effect on voters is noted in multiple real-world elections~\cite{Adams12,AdamsI00}. However, a fundamental computational question in the Hotelling-Downs model remains largely unexplored:

\medskip
\centerline{\em Can equilibria in the Hotelling--Downs model be computed efficiently?}
\medskip

Most prior work instead focuses on the characterization and properties of equilibria, such as the social welfare. A large number of variants and extensions of the model have been studied, including when policy positions are chosen in a multidimensional space~\cite{EatonL75,Plott67}; when voters vote probabilistically~\cite{Coughlin92}; when the set of candidates is not fixed, and candidates can choose to enter~\cite{FedersenSW90,Palfrey84} or leave~\cite{SenguptaS08}; when voting is costly, and voters may not vote if all candidates are far from their position~\cite{CallanderW06,JonesSF22,Smithies41}; or when some candidates are fixed in their positions and cannot strategise~\cite{JonesSF22,Ronayne18}. 
However, despite the interest, the task of computing an equilibrium in the basic Hotelling-Downs model is largely ignored, particularly for continuous variants. 

In this paper, we give both structural results and algorithms for computing equilibria in the basic Hotelling-Downs model. Our results are based on carefully understanding and characterizing the possible local deviations at equilibrium --- when a candidate may change her strategy, keeping the relative order of candidates unchanged --- as well as global deviations --- when a new strategy leads to a change in the relative order of candidates. We present three different algorithms, depending on whether the voter distribution and the set of possible candidate locations are continuous or discrete. 

Formally, the set of voters is distributed in a bounded interval $[0,R]$, with $F$ being the measure of the set, and $\points_v$ being the support. $F(a,b)$ for $a,b \in [0,R]$ is thus the mass (or number) of voters in the interval $[a,b]$. The measure may have atoms (i.e., the measure may be non-zero for a point), and hence $F(a,a)$ is the voter mass at the point $a \in [0,R]$.  
Since the measure $F(a,b)$ may be a real number, when required we assume a model of computation where we can do elementary operations over real numbers in unit time, such as the Blum-Shub-Smale model~\cite{BlumSS89}. We also refer to $\points_v$ as the voter locations. We are further given $\points_c$ as a set of possible locations for $m$ candidates. Our objective in this paper is to decide if an equilibrium exists, and if so, to compute it. For brevity, we refer to this as the problem of \emph{deciding equilibrium existence}.

\subsection{Our Contribution}

Our first contribution is an algorithm that, given an arbitrary distribution of voters and a finite set of $\points_c$ of possible candidate locations, and an $\epsilon \ge 0$, determines in time polynomial in $|\points_c|$ if there exists an $\epsilon$-equilibrium in $\points_c$. Thus given as input a finite set $\points_c$ of possible candidate locations, our algorithm determines existence of an equilibrium in polynomial time. If an equilibrium does not exist, combined with a binary search on $\epsilon$, the algorithm obtains a near-optimal approximate equilibrium as well. This thus comprehensively addresses the problem of deciding equilibrium existence, when the set $\points_c$ is finite.

Our algorithm is based on a carefully designed dynamic program that considers the effect of both local deviations by a candidate --- that preserve the relative order of candidates --- as well as global deviations --- when a deviation leads to a change in the relative order of candidates. The dynamic program iterates over the set of candidates, assigning a position to each. To ensure local deviations do not increase the utility of a candidate, the program simultaneously assigns candidates to the immediate left and right as well. This crucially uses the fact that candidate utilities depend only on their immediate neighbors. To ensure global deviations do not increase the utility of a candidate, the dynamic program tracks the minimum utility of a candidate so far, as well as the largest utility available by a global deviation, ensuring the minimum utility is at least the largest available utility from a global deviation. 

We then consider the case where the set $\points_c$ of possible candidate locations is continuous but the set $\points_v \subset [0,R]$ of voters is discrete but integral.\footnote{If some points in $\points_v$ are not integral, but rational, we can multiply each location by the least common multiple of the denominators to satisfy the assumption, with a corresponding increase in the $R$.} A priori, it is possible that at any equilibrium, some candidate must be at an irrational point. Our main contribution here is a structural result that resolves this concern. We show that if an equilibrium exists, then there exists one in which all candidate positions are rational numbers with denominators bounded by $2^m$, and additionally are very close to the set of integer points in $[0,R]$.

The structural result thus bounds the set of possible candidate locations in two ways. Firstly, it bounds the bit complexity of a candidate location at equilibrium, allowing for a discrete algorithm. Secondly, it bounds the number of such possible locations by $O(mR)$, where $m$ is the number of candidates. Given this structural result, we show that the dynamic program discussed earlier can then be modified to decide the existence of an exact equilibrium in time $\poly(m, R, |\points_v|)$. If $R$ is polynomially bounded, this gives a polynomial time algorithm for deciding equilibrium existence in this case.

Conceptually, the proof of the structural result shows that local deviations by candidates at an equilibrium are bounded by \emph{reflection barriers}, where a reflection barrier is either the reflection of a candidate at a voter location, or a reflection of such a reflection. Thus, from a starting equilibrium, candidate positions can be "snapped" to nearby reflection barriers. We then give an algorithm that systematically carries out this "snapping", yielding our structural result.

 Lastly, we consider the model where $\points_c$ and $\points_v$ are both continuous. Here we require that the voter distribution is nonatomic, and there is a bounded voter density function $f$ over $[0,R]$. In this case, it may be that in all equilibria, some candidate is at an irrational point. We give algorithms for deciding the existence of $\epsilon$-equilibria. Let $f$ be the voter density function, and $M$ be an upper bound on $f$. We show that for any $\epsilon > 0$, if the given instance has an $\epsilon$-approximate equilibrium, then the algorithm returns a $4\epsilon$-approximate equilibrium in time $\poly(m,M,R,1/\varepsilon)$. 
 
 Our algorithm builds upon the algorithm for approximate equilibria from the previous discrete setting, appropriately discretizing the interval $\points_c$, and then running the previous dynamic program. 

For exposition, in the main body of the paper, we make the assumption that the positions occupied by the candidates are \emph{distinct}, i.e., $s_i \neq s_j$ for $i \neq j$. 
In the appendix, we extend all our results to the significantly more complicated setting where multiple candidates may occupy the same location.

\subsection{Related Work} 

We briefly discuss prior work closely related to equilibrium computation in the basic Hotelling-Downs model. As mentioned, many variants and extensions of the basic model have also been studied. A recent survey is presented by Drezner and Eiselt~\cite{DreznerE24}.

Eaton and Lipsey study equilibria in the Hotelling-Downs model in various settings, including on the line and on a circle~\cite{EatonL75}. For the model on a line, they show that for 3 candidates, an equilibrium does not exist under general conditions. They also claim to show necessary and sufficient conditions for a given strategy profile to be an equilibrium. However, the sufficiency conditions turn out to be false~\cite{hayashi2019necessary}. Since exact equilibria may not exist,~\citet{bhaskar2024nearly} study approximate equilibria in the Hotelling-Downs model. They show that for $m>3 $ candidates, there is always a $ \frac{1}{m+1} $-approximate equilibrium, which is simply the strategy profile that places each of the $m$ candidates at the $(m+1)$th quantiles of the voter distribution.~\citet{fournier2019general} also studies approximate equilibria where the voters are distributed in a graphical network and the candidates choose their locations on the graph to maximize their votes, assuming the voters vote for the closest candidate (according to the underlying graph metric). In this model, the author showed that if the number of candidates is large enough, then an approximate equilibrium always exists.

A model closely related to ours is the discrete Voronoi game, introduced by D\"{u}rr and Thang \cite {Durr}. This model considers an undirected graph where vertices represent voters, and edges indicate distances between them. Candidates simultaneously select vertices as their locations, aiming to maximize their vote share, with each voter choosing the nearest candidate. The set of vertices nearest to a candidate $i$ is called $i$'s Voronoi region. 
Determining the existence of a Nash equilibrium in this setup is NP-complete. Necessary and sufficient conditions for equilibrium existence on cycle and path graphs for discrete Voronoi games are known~\cite{enomoto2018pure,mavronicolas2008voronoi}. Discrete Voronoi games on paths are a special case of our model in~\Cref{sec:discrete-candidates}, with successive distances equal to one and  a single voter at each point. 

Equilibrium computation has also been studied in a variant of the Hotelling-Downs model with party nominees. Here, each party has a finite set of nominees, each of which has a fixed position, and must choose a single nominee from among these to contest the election. An equilibrium may not exist even for two parties, and in general finding an equilibrium is NP-hard~\cite{HarrensteinLST21}. There is however an XP algorithm, and this is the best possible, assuming the exponential-time hypothesis~\cite{DeligkasEG22}. If the set of nominees for all parties is the same, then this model is roughly equivalent to the model we study in Section~\ref{sec:discrete-candidates}. For this special case, we give a polynomial time algorithm for equilibrium computation, while the general problem with party nominees is shown to be $\mathcal{NP}$-hard as well as $\mathcal{W}$[1]-hard for certain parameters.

\subsection{Preliminaries}
\label{sec:preliminaries}

As mentioned, the set of voters is distributed in a bounded interval $[0,R]$, with $F$ being the measure of the set, and $\points_v$ being the support. We refer to $\points_v$ as the set of voter locations, and place further restrictions on $\points_v$ in the following sections. When $\points_v$ is continuous, the measure $F(a,b)$ may be a real number. In these case we assume a model of computation where we can do elementary operations over real numbers in unit time, such as the Blum-Shub-Smale model~\cite{BlumSS89}. 

Further, $\points_c$ is the set of possible locations for $m$ candidates. A strategy profile $\strat = (s_1, \ldots, s_m)$ consists of $m$ positions in $\points_c$ for the candidates. In the main body of the paper, for exposition, we assume that the positions occupied by the candidates are \emph{distinct}, i.e., $s_i \neq s_j$ for $i \neq j$. If $\points_c$ is continuous, we will require that $|s_i - s_j| \ge \delta$ for all $i \neq j$ for some small constant $\delta$. It is useful to consider $\delta$ as approaching zero. In the appendix, we show that all our results extend to the setting where multiple candidates may occupy the same location. We note that the two settings differ, even as $\delta$ approaches zero: e.g, with three candidates at distinct locations, one cannot obtain better than a $1/6$-approximate equilibrium. However, allowing multiple candidates at the same location allows us to get a $1/7$-approximate equilibrium~\cite{bhaskar2024nearly}.

For ease of notation, we include two additional candidates that take fixed positions $s_0 = - \infty$ and $s_{m+1} = +\infty$. These are dummy candidates, and for any strategy profile $\strat$ for the remaining candidates, get zero utility, since every voter is closer to a regular candidate than to them. It is easy to verify that their presence does not affect if a strategy profile is an equilibrium. 

For two points $s, s' \in \points_c$, define $\mu(s,s') = (s + s')/2$ as the point where the voters switch from voting for one candidate to the other.

Since each voter votes for the nearest candidate, given a strategy profile $\strat$ and $s \in \strat$, the utility for a candidate at a position $s \in \strat$ is

\[
\util(s,\strat) := \displaystyle F(\mean(s_l,s),\mean(s,s_r)) - \frac{F(\mean(s_l,s),\mean(s_l,s))}{2} - \frac{F(\mean(s,s_r),\mean(s,s_r))}{2} \, ,
\]

\noindent where $s_l$, $s_r \in \strat$ are the positions of the candidates immediately to the left and right of $s$.\footnote{These could possibly be $-\infty$ and $+\infty$ respectively.} The terms subtracted in the expression account for votes at $\mu(s_l,s)$ being evenly split between candidates at $s_l$ and $s$, and votes at $\mu(s,s_r)$ being evenly split between candidates at $s$ and $s_r$. Note that the utility for a candidate depends only on her position and that of her immediate neighbors to the left and right, and not the entire strategy profile.

In Section~\ref{sec:discrete-candidates} since we will iteratively compute a pure Nash equilibrium, we will often need to compute the utility of a candidate in a strategy profile with fewer candidates, thus we will relax the assumption that $|\strat| = m$.

Finally, given points $p$, $q$ in $\points_c$ so that $p < q$, we use $\intr(p,q)$ and $\intu(p,q)$ to denote the best response and best utility for a \emph{new} candidate in the set $(p,q) \cap \points_c$, if there already exist candidates at positions $p$ and $q$, and no other candidates. Thus the only voters that can possibly vote for the new candidate are in the interval $(p,q)$, and their votes may be split between the new candidate and the candidates at $p$ and $q$. Formally, let $\strat' = \{p,q\}$. Then

\[
    \intr(p,q) = \arg \max_{z \in (p,q) \cap \points_c} \util(z,\strat'\cup \{z\}) \quad \text{and} \quad \intu(p,q) = \max_{z \in (p,q) \cap \points_c} \util(z,\strat'\cup \{z\})\, .
\]

Note that if $\points_v$ and $\points_c$ are finite sets, assuming we can compute with real numbers, all of the functions defined above can be computed in polynomial time with respect to $|\points_v|$ and $|\points_c|$. 

\begin{definition}[Pure Nash equilibrium]
    A strategy profile $\strat$ is a \emph{pure Nash equilibrium} (or simply an equilibrium) if for all candidates $i$, and all locations $s_i' \in \points_c$ such that $s_i'\not \in \strat$,\footnote{If $\points_c$ is a continuous set, then we require that $|s_i' - s_j| \ge \delta$ for some small $\delta$ and all $j \neq i$.}  
    \[\util(s_i,\strat) \ge \util(s_i', \strat_{-i}\cup\{s_i'\})\, .\]. 
\end{definition}

Our objective in this paper is to decide if an equilibrium exists, and if so, to compute it. For brevity, we refer to this as the problem of \emph{deciding equilibrium existence}. In~\Cref{sec:bothcontinuous}, when $\points_c = \points_v = [0,R]$, we will be interested in approximate equilibria. Given an $\epsilon \ge 0$, $\strat = (s_1,s_2,...,s_m)$ is an $\epsilon$-equilibrium if for any candidate $i\in[m]$ and any location $s_i'\in[0,1]$ such that $\forall j\neq i$, $|s_i'-s_j|\geq\delta$,  $\util(s_i', \strat_{-i}\cup\{s_i'\})-\util(s_i,\strat))$ $\leq\epsilon$. We will usually think of $\delta$ as approaching zero, and smaller than any of the other parameters in the instance.

\section{Discrete Candidate Positions}
\label{sec:discrete-candidates}

We first consider the case where $\points_c \subseteq \integers_+$ is a finite set where $m$ candidates can choose positions, while the distribution of voters in $[0,R]$ is arbitrary. An instance is thus denoted $(\points_c,m)$, with oracle access to $F$. As mentioned, the measure $F(a,b)$ may be an arbitrary real number, and hence in this section, we assume a model of computation that permits elementary operations on real numbers in unit time. Let $n:= |\points_c|$. We show that the model gives rise to a clean combinatorial structure of equilibria on the line. We leverage this structure to design a dynamic programming algorithm that decides and computes an equilibrium in polynomial time.

The main difficulty lies in handling global deviations by candidates, since dynamic programming naturally captures only local deviations. To address this issue, we provide a structural characterization of equilibria in Proposition~\ref{prop:equilibrium-characterization}.

\begin{restatable}{prop}{equilibriumcharacterization}
\label{prop:equilibrium-characterization}
    Strategy profile $\strat = \{s_1, \ldots, s_m\}$ where $s_1 < s_2 < \ldots < s_m$ is an equilibrium if and only if:\footnote{Recall that we assume there are dummy candidates positioned at $s_0 = -\infty$ and $s_{m+1} = +\infty$.} 
    \begin{enumerate}[label=(\roman*)]
        \item \label{prop.interval} \textbf{Local optimality:} For all $j \in [m]$, $s_j \in \intr(s_{j-1},s_{j+1})$, and
        \item \textbf{Global stability:} $\min_{j \in [m]} u(s_j, \strat) \ge \max_{k \in \{0, \ldots, m\}} \intu(s_k, s_{k+1})$, i.e., the minimum utility for any candidate is at least the maximum utility that can be obtained by a new candidate in any interval between two candidates.
    \end{enumerate}
\end{restatable}
Importantly, this characterization simplifies the verification of equilibrium. Instead of checking all possible unilateral deviations, it suffices to verify that each candidate is locally optimal and that a single global condition holds: the minimum utility among the existing candidates is at least the maximum utility that any additional candidate could obtain by choosing a location on the line.

\begin{proof}[proof of \Cref{prop:equilibrium-characterization}]
For the if direction, for the sake of contradiction, assume that $\strat = \{s_1, \ldots, s_m\}$ is not an equilibrium. Therefore, there exists a candidate $i$ whose best response is to move to another location $s_i' \in \points_c\setminus\strat$, which would strictly increase her utility. If $s_{i-1} < s_i' < s_{i+1}$, this contradicts the first condition of the proposition, since $s_i \in \intr(s_{i-1}, s_{i+1})$. Now, suppose $s_i' \notin [s_{i-1}, s_{i+1}]$. According to the second condition of the proposition, we have $\util(s_i, \strat) \geq \max_{k \in \{0, \ldots, m\}} \intu(s_k, s_{k+1})$, which contradicts the assumption that candidate $i$ can strictly increase her utility by moving to position $s_i'$.

For the other direction, let $\strat$ be an equilibrium and consider an arbitrary candidate $i \in [m]$. Then, $s_i$ must be a best response for candidate $i$, which means $s_i \in \intr(s_{i-1}, s_{i+1})$. Therefore, for all $j\in[m]$, $s_j\in\intr(s_{j-1},s_{j+1})$. Now, let $i_{\min} = \arg \min_{j \in [m]} u(s_j, \strat)$ be a candidate with the minimum utility. Since $\strat = \{s_1, \ldots, s_m\}$ is an equilibrium, $s_{i_{\min}}$ is also a best response for candidate $i_{\min}$. Thus, $\util(s_{i_{\min}}, \strat) \geq \max_{k \in \{0, \ldots, m\}} \intu(s_k, s_{k+1})$, which is the second condition of the proposition.
\end{proof}

The next proposition shows that the number of possible different utility values is bounded by a cubic. Given an instance, $(\points_c,m)$ and an oracle access to $F$, let $\mcV$ be the set of all possible utility values obtainable by a candidate in any strategy profile. Thus $\mcV=\{\util(s,\strat)\mid \strat \text{ is a strategy profile, } s\in\strat\}$.

\begin{restatable}{prop}{differentutilities}
\label{prop:different-utilities}
   Given an instance $(\points_c, m)$, $|\mcV|\le\binom{n+2}{3}=\mathcal{O}(n^3)$, where $n = |\points_c|$.
\end{restatable}

\begin{proof}
    Recall that each voter selects the closest candidate. Hence, to determine the votes received by any candidate $i$, it suffices to know the positions of the nearest candidates to the left and right of $i$. If no such candidates exist, we assume hypothetical candidates located at $-\infty$ and $+\infty$, respectively. This yields a total of $n + 2$ candidate positions. Since the utility of a candidate is determined by the triplet consisting of its position and the positions of its immediate neighbors, the total number of distinct utility values is at most $\binom{n + 2}{3}$.
\end{proof}

We give the complete algorithms in the appendix (\Cref{algo:dp-init} and~\Cref{algo:dp}) and describe the outline here. Our dynamic program works backwards from the last three candidates $m-1$, $m$, and $m+1$ (the last of which is fixed at $+\infty$). For every position of three successive candidates $i-1$, $i$, and $i+1$, we keep track of the minimum utility (termed $\minu$) for candidates $i$, $i+1$, $\ldots$, $m$; as well as the maximum utility (termed $\maxd$) obtainable by a \emph{new} candidate choosing a position after candidate $i-1$. If we can find positions $s_1, \ldots, s_m$ for all the candidates so that $\minu \ge \maxd$ (i.e., the minimum utility for any candidate is at least the maximum utility obtainable by a new candidate) and each candidate $i$ chooses the best position between $s_{i-1}$ and $s_{i+1}$, then by~\Cref{prop:equilibrium-characterization} these positions are an equilibrium.

The dynamic program thus fills up a \emph{binary-valued table} $T(i, s_{i-1}, s_i, s_{i+1}, \minu, \maxd)$ indexed by (i) $i\in[m]$, the current candidate being considered; (ii) $s_{i-1}$, $s_i$, $s_{i+1}\in\points_c$ are positions for candidates $i-1$, $i$, and $i+1$; (iii) $\minu\in\mcV$ is a lower bound on the utility for candidates $i$, $i+1$, $\ldots$, $m$; and (iv) $\maxd\in\mcV$ is an upper bound on the utility obtainable by a new candidate that deviates to a position after $s_{i-1}$.

Note that the size of the table is $m \times (n+2)^3 \times 4n^3 \times 4n^3$, since $\minu$ and $\maxd$ can take at most $4\binom{n+2}{3}$ different values (\Cref{prop:different-utilities}). 

To set the value of $T(i,s_{i-1},s_i,s_{i+1},\minu, \maxd)$, let $\strat' = \{s_{i-1}, s_i, s_{i+1}\}$. Then the table entry is 1 if the following conditions are satisfied:

\begin{enumerate}
    \item[C1:] $\util(s_i, \strat') = \intu(s_{i-1},s_{i+1}) \ge \minu \ge \maxd  \ge \max \{\intu(s_{i-1},s_i), \intu(s_i, s_{i+1})\}$, and
    \item[C2:] there exist $s_{i+2}  \in \points_C$ and $\minu', \maxd' \in \mcV$ so that $s_{i+2}> s_{i+1}$, $\minu' \ge \minu \ge \maxd \ge\maxd'$, and $T(i+1,s_i, s_{i+1}, s_{i+2}$ $, \minu', \maxd') = 1$.
\end{enumerate}

\noindent The entry is 0 otherwise. If $i = m$ (i.e., in the base case), the second condition is not checked. Finally, if there exist $s_1$, $s_2$, $\minu$, $\maxd$ so that $T(1, s_0 = -\infty, s_1, s_2, \minu, \maxd) =1$, the algorithm returns that there is an equilibrium, else it returns that there is no equilibrium. 

Note that given a strategy profile $\strat$, all the required functions --- $\util$, $\intr$, and $\intu$ --- can be computed in polynomial time (in $m$ and $n$). Therefore, we obtain the following theorem.
\begin{restatable}{theorem}{thmdp}
\label{thm:dp}
    Algorithm \ref{algo:dp} decides exact equilibrium existence in polynomial time.
\end{restatable}

It turns out that an $\epsilon$-equilibrium also exhibits a combinatorial structure similar to that of a pure Nash equilibrium. Therefore, given $\epsilon \ge 0$, the algorithm can be easily modified to return an $\epsilon$-equilibrium if it exists. 

\begin{restatable}{theorem}{thmdpapprox}
\label{thm:dpapprox}
    Given $\epsilon \ge 0$, Algorithm \ref{algo:dp} with modified conditions for table-filling decides $\epsilon$-equilibrium existence in polynomial time. 
\end{restatable}

In Appendix~\ref{apndx:discrete}, we show that Algorithm~\ref{algo:dp} can be extended to handle the case where multiple candidates choose the same location. Although the table size increases, the running time remains polynomial in $|\points_c|$ and $m$, and all guarantees continue to hold.

\section{Continuous Candidate and Discrete Voter Positions}
\label{sec:continuous-candidates}

Next, we consider the case where the set $\points_v \subset [0,R]$ of voter positions is discrete, while candidates choose positions in a continuous set. Since $\points_v$ is discrete, we do not need computation on arbitrary reals. Thus in this section, we assume the regular Turing model of computation. We will assume in our results that $\points_v \subseteq \integers_+$. If the points in $\points_v$ are not integral, but are rational, we can multiply each location by the least common multiple of the denominators to satisfy the assumption, with a corresponding increase in the $R$. Let $n := |\points_v|$. We assume the $m$ candidates can choose distinct locations in the interval $\points_c = [0,R]$. For each $p \in \points_v$, $F(p)$ is the number of voters at $p$. Here we give an algorithm that determines existence of an exact equilibrium, in time exponential in $m$ and polynomial in $R$ and $n$. Thus if the number of candidates is a constant, and $R$ is polynomial in $n$, we decide the existence of an exact equilibrium in polynomial time.

Formally, an input instance in this model is $\mcI = (\points_v,m,F,R)$. The $m$ candidates choose as strategies distinct points in $[0,R]$, and the function $F:\points_v \rightarrow \integers_+$ returns the number of voters at each point in $\points_v$.

To decide the existence of an exact equilibrium, the first step is clearly to figure out which are the possible points --- say $\hat{\points}_c$ --- where candidates can choose positions at an equilibrium. If we can obtain a finite set $\hat{\points}_c$ which has the property that if there is an equilibrium, then there exists an equilibrium $\strat$ where $\strat \subseteq \hat{\points}_c$, then we can restrict our search to points in $\hat{\points}_c$, and use the results from~\Cref{sec:discrete-candidates} to decide existence of an exact equilibrium in time polynomial in $|\points_v|$ and $|\hat{\points}_c|$.

A first conjecture for $\hat{\points}_c$ could be that each $p \in \hat{\points}_c$ is either a voter position in $\points_v$, or is a point in the middle of two voter positions. I.e., $\hat{\points}_c \subseteq \points_v \cup \{\mean(x,y) \, : \, x, y \in \points_v\}$. This turns out to be false. \Cref{fig:B-pne} (in~\Cref{sec:B-pne}) gives an instance with a unique equilibrium $\strat$ and a candidate $i$ such that $s_i \not \in \points_v \cup \{\mean(x,y):x, y \in \points_v\}$. Thus, restricting attention to points in $\points_v$ and the mid-points between points in $\points_v$, is not sufficient to obtain an equilibrium. 

In fact, a priori, it is not clear that $\hat{\points}_c$ is a finite, or even a countable set. It may be possible that at any equilibrium, some candidate is at an irrational position. We show algorithmically that in fact this is not the case, and we can obtain a finite set $\hat{\points}_c$. We define $\bit(i) := \{x= a/b \in \rationals_+ \cap [0,R] : \, b = 2^i\}$ as the set of rationals in $[0,R]$ for which the denominator is $2^i$ (and hence can be represented with $i$ bits).

\begin{restatable}{theorem}{lowcomplexity}
\label{thm:lowcomplexity}
If there is an equilibrium for instance $\mcI=(\points_v,m,F,R)$, then there is an equilibrium $\strat$ where each strategy $s \in \strat$ is a rational number in $[0,R]$. Further, each $s \in \bit(m)$.
\end{restatable}
\begin{figure}
    \centering
    \begin{tikzpicture}[line cap=round,line join=round,>=triangle 45,x=1cm,y=1cm,scale=0.58]
    \clip(-2,-3) rectangle (22,4);

\draw [line width=1pt] (21,-1) -- (21,1);
\draw [line width=1pt] (0,-1) -- (0,1);
\draw [line width=1pt] (10.5,-1) -- (10.5,1);

\draw [line width=1pt] (0,0) -- (21,0);

\foreach \x [evaluate=\x as \y using \x*0.34] in {1,...,10} {
    \draw [color = blue, line width=1pt] (\y,-0.3) -- (\y,0.3);
}

\foreach \x [evaluate=\x as \y using 10.5+\x*0.34] in {1,...,10} {
    \draw [color = blue, line width=1pt] (\y,-0.3) -- (\y,0.3);
}
\foreach \x [evaluate=\x as \y using 10.5-\x*0.34] in {1,...,10} {
    \draw [color = blue, line width=1pt] (\y,-0.3) -- (\y,0.3);
}
\foreach \x [evaluate=\x as \y using 21-\x*0.34] in {1,...,10} {
    \draw [color = blue, line width=1pt] (\y,-0.3) -- (\y,0.3);
}

\draw [color = blue, line width=1pt] (3.74,-0.7) -- (3.74,0.7);
\draw [color = blue, line width=1pt] (14.24,-0.7) -- (14.24,0.7);
\draw [color = blue, line width=1pt] (6.76,-0.7) -- (6.76,0.7);
\draw [color = blue, line width=1pt] (17.26,-0.7) -- (17.26,0.7);

\begin{footnotesize}

\draw[color=black] (0,-1.4) node {$T-1$};
\draw[color=black] (21,-1.4) node {$T+1$};
\draw[color=black] (10.5,-1.4) node {$T$};

\draw[color=black] (3.74,1.3) node {$(T-1)+\frac{11}{2^5}$};
\draw[color=black] (14.24,1.3) node {$T + \frac{11}{2^5}$};
\draw[color=black] (6.76,1.3) node {$T-\frac{11}{2^5}$};
\draw[color=black] (17.26,1.3) node {$(T+1)-\frac{11}{2^5}$};
\end{footnotesize}

\end{tikzpicture}
    \caption{The set $\chi\cap[T-1,T+1]$ with five candidates, where $T$ is an integer satisfying $1 \le T \le R-1$.}
    \Description{The set $\chi\cap[T-1,T+1]$ with five candidates, where $T$ is an integer satisfying $1 \le T \le R-1$.}
    \label{fig:chi}
\end{figure}
\Cref{thm:lowcomplexity} thus gives a finite set $\hat{\points}_c$ of size $2^m R$ where we can restrict our search for a possible equilibrium. Our next theorem reduce the dependency on $m$ to a polynomial, showing that in fact we can restrict our attention to locations very close to the integral points in $[0,R]$. We define the following set:

\[
\chi
:=
\left\{
x \pm \frac{k}{2^m}
\,\middle|\,
x \in \{0,\dots,R\},\;
k \in \{0,\dots,2m\}
\right\} \cap [0,R].
\]

The set $\chi$ thus has size $\mathcal{O}(mR)$ and contains the $2m$ points in $\bit(m)$ adjacent to integer points in $[0,R]$. Figure~\ref{fig:chi} illustrates the structure of the set $\chi$ within the interval $[T-1,T+1]$ for an instance with five candidates, where $T$ is an integer satisfying $1 \le T \le R-1$. We then show the following.

\begin{restatable}{theorem}{thmdense}
\label{thm:dense}
If there is an equilibrium for instance $\mcI=(\points_v,m,F,R)$, then there is an equilibrium $\strat$ where each strategy $s \in \strat$ is in $\chi$.
\end{restatable}

\Cref{thm:dense} thus gives us a way to discretize the interval $[0,R]$. We define $\hat{\points}_c = \chi$. 
One could then possibly use the algorithm from ~\Cref{sec:discrete-candidates} to decide the existence of an equilibrium in $\hat{\points}_c$. This, however, may not work: all equilibria on the set $\hat{\points}_c$ may not be equilibria on the set $\points_c$, since candidates can potentially deviate to positions not in $\hat{\points}_c$. We need to modify the dynamic programming algorithm from~\Cref{sec:discrete-candidates} to find an equilibrium where the candidates are at points in $\hat{\points}_c$, but \emph{do not have a better deviation in $\points_c$}. 

For this, we modify the dynamic programming algorithm from~\Cref{sec:discrete-candidates} as follows: in Condition C1 when computing $\intu(s_{i-1}, s_{i+1})$, we search over the continuous interval $(s_{i-1}, s_{i+1}) \cap \points_c$ instead of the discrete set $(s_{i-1}, s_{i+1}) \cap \hat{\points}c$. By Lemma~\ref{lm:best-u}, this can be done in polynomial time. Observe that this modified dynamic programming returns an equilibrium where the candidates are located in $\hat{\points}_c$, but have no better deviation in $\points_c$ (and is hence an equilibrium in the given instance). The algorithm runs in time $\poly(m,|\points_v|,R)$. 

\begin{restatable}[]{theorem}{lowcomplexityalgo}
    For instance $\mcI=(m,\points_v,F,R)$, we can decide existence of an equilibrium in time $\mathcal{O}(m|\chi|^{16}|\points_v|^2)$ where $|\chi|=\mathcal{O}(mR)$.
\label{thm:lowcomplexityalgo}
\end{restatable}

\subsection{Equilibria with low bit complexity}

We now present intuition for the proof of~\Cref{thm:lowcomplexity}. We start with an arbitrary equilibrium $\strat$ in the given instance, where some candidates may use strategies $s \in \strat$ so that $s$ is either irrational, or $s \not \in \bit(m)$. We will then slightly shift the positions of these candidates, to obtain an equilibrium strategy profile $\strat'$ where each strategy satisfies the conditions in the theorem. Throughout the process of shifting the candidates, we ensure that the strategy profile is an equilibrium, and the order of the candidates remains unchanged.

To prove~\Cref{thm:lowcomplexity}, recall that our process is to start with an arbitrary equilibrium $\strat$, and shift the positions of these candidates, to obtain an equilibrium strategy profile $\strat'$ where each strategy satisfies the conditions in the theorem. The main technical question we then need to address is, \emph{what are the barriers to shifting candidates}? That is, what are the boundaries \emph{beyond} which if we shift a candidate while maintaining the order of candidates, the strategy profile is no longer an equilibrium? We address this in the following motivating example.

\begin{example}\label{example:01}
Figure \ref{fig:shifting1} represents a PNE with 5 candidates and 24 voters. The red circles indicate the voters' locations, with the numbers below each circle representing the number of voters at that location. The blue numbers correspond to the lengths of the respective line segments. 

The figure shows an equilibrium strategy profile $\strat = \{0,2,9,16,18\}$ with the candidates shown as black triangles. We verify in the appendix that this is indeed an equilibrium, with their utilities being 5, 5, 4, 5, and 5 in order of their position from left to right. 

Suppose we start shifting candidate 2 to the right. If she crosses the dashed line marked as point $c$, her distance from the voters at $b$ is larger than $2$, while candidate 1 is at distance $2$ from these voters. Hence if she crosses the dashed line, her utility goes down to $0$, and this is no longer an equilibrium. The point $c$ for the dashed line is the reflection of $s_1$ in $b$, i.e., $c = b + |b - s_1|$.

\end{example}

\begin{figure}[ht]
 \centering
    \begin{tikzpicture}[line cap=round,line join=round,>=triangle 45,x=1cm,y=1cm,scale=0.6]

\draw [line width=1pt] (0,0) -- (18,0);

\draw [line width=1pt,color=blue,dashed] (4,-0.7) -- (4,0.7);
\draw [line width=1pt,color=blue,dashed] (8,-1.3) -- (8,1.3);
\draw [line width=1pt,color=blue,dashed] (10,-1.3) -- (10,1.3);
\draw [line width=1pt,color=blue, dashed] (14,-0.7) -- (14,0.7);

\begin{footnotesize}
\draw [fill=red] (0,0) circle (6pt);
\draw [fill=red] (2,0) circle (6pt);
\draw [fill=red] (7,0) circle (6pt);
\draw [fill=red] (11,0) circle (6pt);
\draw [fill=red] (18,0) circle (6pt);
\draw [fill=red] (16,0) circle (6pt);

\node[isosceles triangle,
    draw,
    rotate=90,
    fill=black,
    minimum size =0.05cm, scale=0.5] (T1)at (0,1.3){};

\node[isosceles triangle,
    draw,
    rotate=90,
    fill=black,
    minimum size =0.05cm, scale=0.5] (T2)at (2,1.3){};
\node[isosceles triangle,
    draw,
    rotate=90,
    fill=black,
    minimum size =0.05cm, scale=0.5] (T3)at (9,1.3){};
\node[isosceles triangle,
    draw,
    rotate=90,
    fill=black,
    minimum size =0.05cm, scale=0.5] (T4)at (16,1.3){};
\node[isosceles triangle,
    draw,
    rotate=90,
    fill=black,
    minimum size =0.05cm, scale=0.5] (T5)at (18,1.3){};

\end{footnotesize}
\draw[color=blue] (1,0.5) node {$2$};
\draw[color=blue] (3,0.5) node {$2$};
\draw[color=blue] (5.5,0.5) node {$3$};
\draw[color=blue] (7.5,0.5) node {$1$};
\draw[color=blue] (9,0.5) node {$2$};
\draw[color=blue] (10.5,0.5) node {$1$};
\draw[color=blue] (12.5,0.5) node {$3$};
\draw[color=blue] (15,0.5) node {$2$};
\draw[color=blue] (17,0.5) node {$2$};

\draw[color=black] (0,2) node {$s_1$};
\draw[color=black] (2,2) node {$s_2$};
\draw[color=black] (9,2) node {$s_3$};
\draw[color=black] (16,2) node {$s_4$};
\draw[color=black] (18,2) node {$s_5$};

\foreach \x in {0,2,16,18} {
        \draw (\x,-0.8) node[circle, inner sep=1pt, draw, color=red, fill=white] {5}; 
    }

\foreach \x in {7,11} {
        \draw (\x,-0.8) node[circle, inner sep=1pt, draw, color=red, fill=white] {2}; 
    }

\foreach \x/\y in {0/a,2/b,4/c,7/d,8/e,9/f,10/g,11/h,14/i,16/j,18/k} {
\draw[color=black] (\x,-1.8) node {$\y$};
}
    
\end{tikzpicture}
    \caption{An equilibrium with 5 candidates and 24 voters. The red circles indicate the voters' locations, with the numbers below each circle representing the number of voters at that location. The blue numbers correspond to the lengths of the respective line segments. At any equilibrium, the middle candidate must remain between the two inner dashed lines.}
    \Description{An equilibrium with 5 candidates and 24 voters. The red circles indicate the voters' locations, with the numbers below each circle representing the number of voters at that location. The blue numbers correspond to the lengths of the respective line segments. At any equilibrium, the middle candidate must remain between the two inner dashed lines.}
    \label{fig:shifting1}
\end{figure}

The example suggests that one of the barriers we must consider when shifting candidates to preserve an equilibrium is the \emph{reflections} of candidates in voter positions.

Formally, define function $f_r:\mathbb{R}^2\rightarrow \mathbb{R}$ as $f_r(x,p)=2p-x$. Hence $f_r(x,p)$ returns the reflection of $x$ in $p$. We then define the following set for a location $x\in[0,R]$:

\begin{align*}
B(x, \points_v) &= \{f_r(x,p)~|~p\in \points_v\} \, .
\end{align*}

\noindent Given a candidate's location $s_i$, $B(s_i,\points_v)$ is the set of reflections of $s_i$ in all voter positions. Suppose $f_r(s_{i-1},p)=x$, for some $p \in \points_v$. Then, if $s_i < x$, voters at $p$ vote for $s_i$; if $s_i > x$, voters at $p$ vote for $s_{i-1}$; and if $x=s_i$, then voters at $p$ split their votes equally among $s_{i-1}$ and $s_i$. 

The next lemma shows that given an equilibrium $\strat$, if we shift a candidate $i$ to the right (i.e., increase $s_i$) without crossing any points in $\points_v$ and $B(s_{i-1}, \points_v)$, the utility of each candidate remains unchanged. Note that the resulting strategy profile may not be an equilibrium.

\begin{restatable}{lemma}{lmone}
\label{lm:1}
    Let $\strat$ be a PNE. Fix candidate $i$, and let $A= B(s_{i-1},\points_v)\cup \points_v$ be the set of voter positions and the reflections of $s_{i-1}$.  If $s_i \not \in A$, let $x$ be such that $s_i<x<s_{i+1}$ and $A\cap[s_i,x]=\emptyset$, and let $\strat'=\{s_1,\dots,s_{i-1},x,s_{i+1},\dots,s_m\}$. Then each candidate has the same utility in $\strat$ and $\strat'$. 
\end{restatable}

Are the voter positions and these reflections the only barriers we need to keep in mind while shifting candidates, if we want to maintain an equilibrium? The above example in fact shows that this is not the case, and we need to look at \emph{reflections of reflections} as well.

\begin{example}\label{example:02}

Continuing with the example shown in Figure \ref{fig:shifting1}, now let us consider what happens as we shift candidate 3 to the right. We will show that if $s_3$ crosses the dashed line at $g$, the strategy profile is no longer an equilibrium. Note that $g$ is not the reflection of any candidate in any voter, i.e., $g \not \in B(s,\points_v)$ for any $s \in \strat$.

To see this, note that if candidate 3 shifts to a position $g + \delta$, her distance from the voters at $d$ --- that were voting for her in $\strat$ --- is $3+\delta$. But this is not an equilibrium, since now candidate 2 can increase her vote share. Candidate 2 can shift to the position $c- \delta/2$. Then her distance to $b$ is $2 - \delta/2$, and her distance to $d$ is $3 + \delta/2$. She is then the closest candidate to both these positions, and her vote share increases from 5 to 7.
\end{example}

What makes the position $g$ important in the example is that it is the reflection of a reflection. In particular, point $c$ is the reflection of $s_1$ in $b \in \points_v$, and $g$ is the reflection of $c$ in $d \in \points_v$. For candidate 3, staying to the left of $g$ thus prevents candidate 2 from moving to position $c$ and snatching the votes from the voters in $d$. We now introduce the following sets of points to address these double reflections.

\begin{align*}
B^2(x,\points_v) &= \{f_r(f_r(x,p),q) \mid ~p,q \in \points_v, ~ x < p, ~f_r(x,p) < q, \text{ OR } x > p, ~f_r(x,p) > q\} \, .
\end{align*}

Note that the definition includes all second reflections of a point that are \emph{in the same direction} as the first reflection. That is, if $z \in B^2(x, \points_v)$, then there exist $p$, $q$ in $\points_v$ so that either $q < p$ and $z = f_r(f_r(x,p),q) < f_r(x,p) < x$, or $p < q $ and $x < f_r(x,p) < f_r(f_r(x,p),q) = z$. 

The next lemma shows that given a strategy profile $\strat$ and a candidate $i$, shifting $i$ to the right so that she does not cross any point $\points_v$, any reflection of $s_{i-1}$ in $\points_v$, and any reflection of the reflection of $s_{i-1}$, $s_{i-2}$ in $\points_v$, maintains an equilibrium. This also implies that while shifting a candidate $i$ to the right, one is only concerned with reflections of candidates $i-1$ and $i-2$ to the left. Similarly, while shifting a candidate $i$ to the left, one is only concerned with reflections of candidates $i+1$ and $i+2$ to the right.

\begin{restatable}{lemma}{shifting}
\label{lm:3}
    Let $\strat$ be a PNE, and $i$ be a candidate such that $s_i\notin A$ where $A=B(s_{i-1},\points_v)\cup B^2(s_{i-1},\points_v)\cup B^2(s_{i-2},\points_v)\cup \points_v$. Let $x$ be a location such that, $s_{i}<x<s_{i+1}$ and $A\cap [s_i,x]=\emptyset$. Then, $\strat'=\{s_1,...,s_{i-1},x,s_{i+1},...,s_m\}$ is a PNE.

    Similarly, let $i$ be a candidate such that $s_i\notin A$ where $A=B(s_{i+1},\points_v)\cup B^2(s_{i+1},\points_v)\cup B^2(s_{i+2},\points_v)\cup \points_v$. Let $x$ be a location such that, $s_{i-1}<x<s_{i}$ and $A\cap [x,s_i]=\emptyset$. Then, $\strat'=\{s_1,...,s_{i-1},x,s_{i+1},...,s_m\}$ is a PNE.
\end{restatable}

\subsection*{An algorithm to shift candidates at equilibrium}


Given an equilibrium $\strat$, we now use~\Cref{lm:3} to construct another PNE $\strat'$ where each candidate's location is in $\bit(m)$. This construction proceeds iteratively, given by~\Cref{algo:1}, in two phases per iteration.

In iteration $i$, the first phase (Algorithm \ref{algo:1}) attempts to identify the leftmost candidate whose location is not in $\bit(i-1)$, say candidate $j$. It then tries to shift candidate $j$ \emph{rightwards} to the next position in $\bit(i)$, say $x$. Note that all candidates to the left of $s_j$ are at positions in $\bit(i-1)$, hence their boundaries do not lie in $(s_j, x)$ (recall that $\points_v \subseteq \mathbb{Z}_+$). However if $x \in \bit(i-1)$, then a boundary for a candidate on the left could lie on $x$. Thus if $x \in \bit(i) \setminus \bit(i-1)$, we shift $j$ to $x$ and move on to the next iteration. If however $x \in \bit(i-1)$, we then enter the second phase and try and shift candidate $j$ \emph{left}.

In the second phase (Algorithm \ref{algo:2}), we know the closest point $x \in \bit(i)$ after $s_j$ is actually in $\bit(i-1)$. Hence the closest point $x' \in \bit(i)$ \emph{before} $s_j$ is not in $\bit(i-1)$. So we try and shift $s_j$ left, towards $x'$. If a candidate's boundary however, lies between $x'$ and $s_j$, clearly this candidate is to the right of candidate $j$, and is in a position not in $\bit(i-1)$. We pick this candidate as the new candidate $j$, and again try and shift it leftwards.

\begin{algorithm}[ht]
    \caption{Shift-Right($m,\strat,\points_v$)}\label{algo:1}
    \begin{algorithmic}[1]
        \State $i=1$
        \While {$i\le m$}
            \State $k=\min\left(\{j\in[m]~|~s_j\notin \bit(i-1)\}\right)$ \Comment{Pick leftmost candidate not at $\bit(i-1)$}
            \If{$k=$ Null} \label{line:knull}
                \State \textbf{Return} $\strat$
            \EndIf
            \State $x=\min(\{y\in \bit(i-1)~|~ y>s_k\})$, $\epsilon = 2^{-i}$ \Comment{$x-\epsilon$ is closest pt in $\bit(i) \setminus \bit(i-1)$ to $s_k$ }
          \If{$s_k \le s_{k+1} \le (x-\epsilon)$} \Comment{Both $s_k$ and $s_{k+1}$ are before $x-\epsilon$}
                    \State $s_{k+1}=(x-\frac{\epsilon}{2})$, $s_k = x-\epsilon$  \Comment{Shift both} \label{line:shiftboth}
            \ElsIf{$s_k \le (x-\epsilon)$}
                \State $s_k = x-\epsilon$         
           \ElsIf{$s_k > (x-\epsilon)$} \Comment{$x-\epsilon$ is to the left of $s_k$}
                \State $\strat =$ Shift-Left$(k,x,\epsilon,i, m,\strat,\points_v)$
            \EndIf
            \State $i = i + 1$
        \EndWhile
        \State \textbf{Return} $\strat$
    \end{algorithmic}
\end{algorithm}

\begin{algorithm}[t]
    \caption{Shift-Left($k,x,\epsilon,i, m,\strat,\points_v$)}\label{algo:2}
    \begin{algorithmic}[1]
        \State $j=k$
        \While{True}
            \State $z = \max(\{y\in B^2(s_{j+1},\points_v)\cup B^2(s_{j+2},\points_v)~|~(x-\epsilon) \le y \le s_{j}\})$ \Comment{$z$ prevents $s_j$ from shifting left to $x-\epsilon$}
            \If{$z=$ Null}
                \If{$s_{j-1}\ge (x-\epsilon)$}
                    \State $s_{j-1}=(x-\epsilon)$
                \Else
                    \State $s_{j} = (x-\epsilon)$
                \EndIf
                \State \textbf{Break}
            \EndIf
            
            \If{$z\in B^2(s_{j+1},\points_v)$} \Comment{$z$ is a candidate boundary for $s_{j+1}$}
                \State $x = \min\left(\{y\in \bit(i-1)\mid y>s_{j+1}\}\right)$ \Comment{$y-\epsilon$ is closest point in $\bit(i) \setminus \bit(i-1)$ to $s_{j+1}$, we prove $(y-\epsilon) < s_{j+1}$.}
                \State $j = j+1$ \Comment{Iterate with $j = j+1$} \label{line:choosej1}
            \ElsIf{$z\in B^2(s_{j+2},\points_v)$}
                \State $x = \min\left(\{y\in \bit(i-1)\mid y>s_{j+2}\}\right)$ \Comment{$y-\epsilon$ is closest point in $\bit(i) \setminus \bit(i-1)$ to $s_{j+2}$, we prove $(y-\epsilon) < s_{j+2}$.}
                \State $j = j+2$ \Comment{Iterate with $j = j+2$} \label{line:choosej2}
            \EndIf
        \EndWhile
        \State \textbf{Return} $\strat$
    \end{algorithmic}
\end{algorithm}

Given~\Cref{lm:3}, the major part of the proof of~\Cref{thm:lowcomplexity} is proving the correctness of~\Cref{line:shiftboth} in~\Cref{algo:1} (i.e., shifting both candidates $s_k$ and $s_{k+1}$ preserves equilibrium), and the correctness of~\Cref{algo:2}. For~\Cref{algo:2}, there are a number of claims and lemmas that go into the proof, including showing that each successive candidate chosen in the while loop (in Lines~\ref{line:choosej1} and~\ref{line:choosej2}) must be shifted to the left.  

By the end of the $m$th iteration (or when $k = $ Null in~\Cref{line:knull}), all candidates will occupy locations that are in $\bit(m)$, and the resulting configuration will still be an equilibrium.

\subsection{Reducing the strategy space}

We now prove~\Cref{thm:dense}, further reducing the search space for possible candidate locations at equilibria. For the proof, for $i \in \{0, \ldots, m\}$, we define the set
\[
\chi_i
=
\left\{
x \pm \frac{k}{2^m}
\,\middle|\,
x \in \{0,\dots,R\},\;
k \in \{0,\dots,2i\} 
\right\}\cap [0,R]
\]

\noindent as the set consisting of all integers in $[0,R]$ together with, for each integer in $(0,R)$, the $4i$ adjacent points in $\bit(m)$, and for the boundary integers $0$ and $R$, the $2i$ adjacent points that lie within $[0,R]$. Note that $\points_v \subseteq \chi_0$ $\subset \ldots \subset \chi_m$, and that we have defined $\chi$ as $\chi_m$. To prove~\Cref{thm:dense}, we first show a result about reflections of points in $\chi_i$. The proof is immediate but stated explicitly for later reference. Recall our assumption that points in $\points_v$ are integral, and that $f_r(x,p)=2p-x$ is the reflection of $x$ in $p$.

\begin{restatable}{prop}{proppolym}
\label{prop:polym}
Let $z \in \chi_i$ for $i \in \{0,\dots,m\}$, and $p \in \points_v$. Then the reflection $f_r(z,p)$ is either outside $[0,R]$, or is also in $\chi_i$.

\end{restatable}

\begin{proof}
Let $z = x \pm \frac{k}{2^m}$ for $x \in \{0, \ldots, R\}$ and $k \in \{0, \ldots, 2i\}$. Since $p \in \mathbb{Z}_+$, $f_r(z,p) = 2p-z$ $=2p-x \pm \frac{k}{2^m}$, which is either outside $[0,R]$ or again of the form $x' \pm \frac{k}{2^m}$ for $x' \in \{0,\ldots,R\}$ and $k$ unchanged.
\end{proof}

\begin{proof}[Proof of Theorem~\ref{thm:dense}]
By Theorem~\ref{thm:lowcomplexity}, if an equilibrium exists for the instance 
$\mcI = (\points_v, m, F, R)$, then there exists an equilibrium profile 
$\strat$ in which every strategy $s \in \strat$ belongs to $\bit(m)$. 
Accordingly, we begin with such an equilibrium $\strat$, where all candidates are located at points in $\bit(m)$, and use it as input to Algorithm~\ref{algo:chishifting}.

    \begin{algorithm}[ht]
    \caption{$\chi$-shifting($m,\strat,\points_v$)}\label{algo:chishifting}
    \begin{algorithmic}[1]

        \For{$i$ from $1$ to $m$}
            
            \If{$s_i \not \in \chi_i$}\label{line:notclose}
                \State $x=$ closest integral point larger than $s_i$.
                \If{$s_{i+1}\in(s_i, x-\frac{2i}{2^m}]$}\label{line:leftcandidate}
                    \State $s_{i+1} = x - \frac{2i-1}{2^m}$
                \EndIf
                \State $s_i = x - \frac{2i}{2^m}$
            \EndIf
        \EndFor
        \State return $\strat$
    \end{algorithmic}
\end{algorithm}

\begin{restatable}{clm}{clmchishifting}
    \label{clm:chishifting}
    Algorithm \ref{algo:chishifting} returns an equilibrium where all the candidates are on points in $\chi$. 
\end{restatable}

\begin{proof}
We prove the claim by induction on the iterations of Algorithm~\ref{algo:chishifting}. 
The key ingredient is Lemma~\ref{lm:3}, which implies that, in order to shift a candidate to the right while preserving equilibrium, it suffices to ensure that the candidate does not cross any voter location, nor any reflection or double reflection generated by the two immediate candidates to her left.

\paragraph{Inductive hypothesis.}
After the $i$-th iteration, for $j \le i$, candidate $s_j$ is in $\chi_j$. Moreover, the resulting strategy profile is an equilibrium.

\paragraph{Base case ($i=1$).}
Consider candidate $s_1$. If $s_1 \in \chi_1$, there is nothing to prove. Otherwise, let $x$ be the smallest integer strictly greater than $s_1$. Since $s_1 \in \bit(m)$, we must have
\[
s_1 < x - \frac{2}{2^m}.
\]

Note that there are no voters in the interval $[s_1, x-\frac{1}{2^m}]$, and hence there can be at most one other candidate (besides $s_1$) in this interval, otherwise there would be a candidate getting zero votes.

If there is no candidate in the interval $(s_1, x-\frac{2}{2^m}]$, the algorithm shifts $s_1$ to $x-\frac{2}{2^m}$ which is in $\chi_1$. Since there are no candidates to the left and points in $\points_v$ are integral, we do not cross any reflection barriers or voters. By Lemma~\ref{lm:3}, the resulting strategy profile is an equilibrium. 

Otherwise, $s_2$ lies in the interval $(s_1, x-\frac{2}{2^m}]$. The algorithm shifts $s_{2}$ to $x-\frac{1}{2^m}$ and then shifts $s_1$ to $x-\frac{2}{2^m}$, both of which are in $\chi_1$. Shifting $s_2$ preserves equilibrium since there are no integers between $s_1$ and $s_2$ (and hence no reflection barriers or voters between $s_2$ and $x-\frac{1}{2^m}$). Shifting $s_1$ preserves equilibrium since there are no candidates to the left.

\paragraph{Inductive step.}
Assume the claim holds through iteration $i-1$. Consider iteration $i$ and candidate $s_{i}$. If $s_{i}$ is already in $\chi_i$, there is nothing to prove. Otherwise, let $x$ be the smallest integer strictly greater than $s_{i}$. Since $s_{i} \in \bit(m)$,
\[
s_{i} < x - \frac{2i}{2^m}.
\]

As noted in the base case, there are no voters in the interval $[s_i, x-\frac{2i}{2^m}]$, and hence there can be at most one other candidate (besides $s_i$) in this interval, otherwise there would be a candidate getting zero votes.

If no candidate lies in $(s_i,x-\frac{2i}{2^m}]$, the algorithm shifts $s_{i}$ to $x-\frac{2i}{2^m}$ which is in $\chi_i$. Otherwise, it shifts $s_{i+1}$ to $x-\frac{2i-1}{2^m}$ and then shifts $s_{i}$ to $x-\frac{2i}{2^m}$. Both these points lie in $\chi_{i}$.

By the inductive hypothesis, every candidate $s_j$ for $j<i$ is in $\chi_{i-1}$. Therefore, by~\Cref{prop:polym}, all points in
\[
A = B(s_{i-1},\points_v)
\cup B^2(s_{i-1},\points_v)
\cup B^2(s_{i-2},\points_v)
\cup \points_v
\]
are also in $\chi_{i-1}$, and hence,
\[
A \cap \left[s_i, x-\frac{2i-1}{2^m}\right] = \emptyset.
\]

Since no point of $A$ is crossed during the shift, Lemma \ref{lm:3} implies that the updated strategy profile remains a pure Nash equilibrium. This completes the induction.
\end{proof}
Hence from Claim \ref{clm:chishifting}, we conclude the proof of Theorem \ref{thm:dense}.
\end{proof}

In Appendix~\ref{apndx:con-dis-mul}, we show that Theorem~\ref{thm:lowcomplexity} and Theorem~\ref{thm:dense} continue to hold even when candidates are allowed to choose the same location. Since a suitable modification of Algorithm~\ref{algo:dp} also works when multiple candidates occupy a single location, all the positive results in this section extend to this setting as well.

\subsection{Bounds on approximate equilibria}

\citet{bhaskar2024nearly} show that (after rescaling so that their is a unit mass of total votes) there always exists a $\frac{1}{m+1}$-approximate equilibrium in which the $m$ candidates are located at the $(m+1)$th quantiles of the voter distribution. In particular, when $\points_v$ is a discrete set, the quantiles correspond to voter locations, and hence there exists a $\frac{1}{m+1}$-approximate equilibrium where all candidates are at voter locations. Given the results on exact equilibria, a natural question is whether for a reduced search space is sufficient for finding approximate equilibria. We show that restricting the search space to the points in $\points_v$ cannot yield a better approximation than $\frac{1}{m+3}$. In the appendix, we give an example where for every $m$, there exists an instance that admits a pure equilibrium, yet when candidates are restricted to voter locations, no equilibrium better than a $\frac{1}{m+3}$-approximation exists.

\begin{restatable}{theorem}{thmappximp}
    \label{thm:appximp}
For every $m$, there exists an instance with a pure Nash equilibrium such that any strategy profile where candidates are restricted to voter locations is a $\frac{1}{m+3}$-equilibrium.
\end{restatable}
This shows that obtaining better than a $\frac{1}{m+3}$-approximate equilibrium necessitates searching beyond voter locations.

\section{Continuous Candidate and Voter Positions}
\label{sec:bothcontinuous}

Lastly, we consider the case where the sets $\points_c = \points_v  = [0,R]$. In this section, we assume that the distribution of voters is nonatomic, and there exists a bounded voter density function $f:[0,R]\rightarrow\mathbb{R^+}\cup \{0\}$ such that $f(z) \le M$ for some finite $M$, and $\int_0^R f(z) \, dz = 1$. Let $F(y) = \int_0^y f(z) \, dz$ be the cumulative distribution function. Since $f(z) \le M$, for any interval of length $\Delta$, the total voters in the interval is at most $M \Delta$. 

In comparison to the previous sections, there are instances where all \emph{exact} equilibria may place candidates at irrational points. As a simple example, consider the instance where the pdf for the voters is $f(x) = 2x$ in the interval $[0,1]$, and there are two candidates. Then at an exact equilibrium, both candidates must be at the median, at $x = 1/\sqrt{2}$.

\paragraph*{Algorithms for approximate equilibria.} We focus on computing approximate equilibria in this case. As mentioned earlier, in the main body of the paper we restrict the candidates to occupy distinct positions in the interval in any strategy profile $\strat$, so that $\min_{i,j} |s_i - s_j|$ $\ge \delta$ for some small $\delta$ (and remove this restriction in the appendix). It is useful to consider $\delta$ as approaching zero. Specifically, we will assume $M \delta < 10^{-3}$. 

For our algorithms to compute approximate equilibria, in order to access the voter density function $f$, we will assume access to an oracle that supports the following queries:
\begin{itemize}
    \item $F(z)$: Returns $F(z)$, the total voters in the interval $[0,z]$.
    \item Cut($z,v$): Given a location $z \in [0,R]$ and a value $v \in [0,R]$, returns a location $y \ge z$ so that $F(y) - F(z) = v$, or returns $1$ if there is no such $y$. If there are multiple such locations, return one arbitrarily but consistently.
\end{itemize}

Since the oracle output may be an arbitrary real number, we assume a model of computation where we can do elementary operations over real numbers in unit time. Our main result in this section is the following.

\begin{theorem}
In the model with continuous voter and candidate locations, if there exists an $\epsilon$-equilibrium for $\epsilon > 0$, then a $4\epsilon$-equilibrium can be computed in time $\poly(m,R, M/\epsilon)$.
\label{thm:main-cvcc}
\end{theorem}

We define the point $q_{i,m+1}$ as the $i$th $(m+1)$-quantile. That is, $F(q_{i,m+1})=i/(m+1)$. Then the strategy profile that places the $m$ candidates at the $(m+1)$th quantiles (i.e., for all $i\in[m]$, $s_i = q_{i,m+1}$) is a $1/(m+1)$-approximate equilibrium~\cite{bhaskar2024nearly}.
Hence if $\epsilon \ge 1/(4m)$, then a $4 \epsilon$-equilibrium is easily obtained. In the remainder, we will assume that $\epsilon \le 1/(4m)$.

Let $\alpha = \epsilon/(4M)$. We assume $1/\alpha$ is an integer. Define $\points_{\alpha} := \{0,\alpha, 2 \alpha, 3 \alpha, \ldots, R\}$. We call these points the $\alpha$-positions, and an interval $(k \alpha, (k+1)\alpha]$ for $k \in [R/\alpha]$ an $\alpha$-interval. We will show that if there exists an $\epsilon$-equilibrium, then there exists a $2\epsilon$-equilibrium where the candidates are located at $\alpha$-positions. Note that given $\epsilon$, $M$, and $R$, determining the $\alpha$-positions does not require any oracle queries. 

Let $\strat$ be an $\epsilon$-equilibrium. We first show a bound on the number of candidates in an $\alpha$-interval. Recall that by assumption, $\epsilon \le 1/(4m)$.

\begin{restatable}{clm}{alphainterval}
   There are at most two candidates in any $\alpha$-interval in $\strat$. If there are two candidates in an $\alpha$-interval, then there are no candidates in the preceding and succeeding $\alpha$-intervals.
    \label{clm:alpha-interval}
\end{restatable}

Given an $\epsilon$-equilibrium $\strat$, we now obtain a $2\epsilon$-equilibrium where the candidates are located at $\alpha$-positions. For this, in each $\alpha$-interval, if there is a single candidate, this candidate is moved to the next $\alpha$-position. If there are two candidates, the candidate on the left is moved to the $\alpha$-position to the left, and the candidate on the right is moved to the $\alpha$-position to the right. By the claim, there is at most one candidate at each $\alpha$-position, and the relative order of the candidates is maintained. Let $\strat' = (s_1', \ldots, s_m')$ be the resulting strategy profile.

\begin{restatable}{clm}{alphadiscretize}
    $\strat' \subseteq \points_\alpha$ is a $2 \epsilon$-equilibrium.
    \label{clm:alpha-discretize}
\end{restatable}

We now use the dynamic program from Theorem \ref{thm:dp} and the properties shown to complete the proof of \Cref{thm:main-cvcc}.

\begin{proof}[Proof of~\Cref{thm:main-cvcc}.]
Assume we are given $\epsilon$. If $\epsilon \ge 1/(4m)$, then as we noted the strategy profile $\strat$ where $s_i = q_{i,m+1}$ is a $1/(m+1)$ equilibrium, and hence is a $4 \epsilon$-equilibrium as required. If $\epsilon < 1/(4m)$, we define $\alpha = \epsilon / 4M$, and obtain the $\alpha$-positions as shown. Now let $\mathcal{I}$ denote the original instance, with the set of possible candidate locations $\points_c = [0,R]$, and $\hat{\mcI}$ denote the discretized instance, with the set of possible candidate locations $\hat{\points}_c = \points_\alpha$. Then~\Cref{clm:alpha-discretize} shows that there exists a $2\epsilon$-equilibrium $\strat$ for $\mcI$, where $\strat \subseteq \hat{\points_c}$. Clearly, $\strat$ is also an $2\epsilon$-equilibrium for $\hat{\mcI}$. We then use the dynamic program from Theorem~\ref{thm:dpapprox} to obtain a $2\epsilon$-equilibrium $\hat{\strat}$ for the instance $\hat{\mcI}$, in time $\poly(m,R,M/\epsilon)$. But since in any $\alpha$-interval, the mass of voters is at most $\epsilon/4$, this gives us an $2\epsilon + \epsilon/4 \le 3 \epsilon$-approximate equilibrium for $\mcI$ as well. Note that the dynamic program does not require Cut queries.

If we are not given $\epsilon$, we can do a binary search. We start with $\epsilon = 1/(8m)$ and run the above procedure. If the dynamic program fails to produce a $3 \epsilon$-equilibrium, then again from~\Cref{clm:alpha-discretize}, there is no $\epsilon$-equilibrium, and we simply return $\strat = (q_{1,m+1}, \ldots, q_{m,m+1})$. Otherwise, if the dynamic program succeeds, we successively multiply $\epsilon$ by $3/4$ until it fails to return a $3\epsilon$-equilibrium. At this point, in the previous iteration, our algorithm has returned a $4 \epsilon$-equilibrium, and hence there exists a $4 \epsilon/3$-equilibrium, but not an $\epsilon$-equilibrium.
\end{proof}

In Appendix~\ref{apndx:con-con-mul}, we prove that Theorem~\ref{thm:main-cvcc} remains valid even when candidates are permitted to choose the same location. Therefore, all the results of this section extend to this setting as well.

\bibliographystyle{plain}

\bibliography{voting}
\newpage


\appendix

\section*{Appendix}

We will use the following additional notation in the appendix. Given a strategy profile $\strat$ and a point $s \in \points_c \setminus \strat$, we define $\utilout(s,\strat)$ as the utility of a \emph{new} candidate at $s$.

\[
 \utilout(s,\strat) = \util(s,\strat \cup \{s\}) \, .
\]

The purpose of this notation is to emphasize that the point $s \not \in \strat$.

\section{Proofs from Section~\ref{sec:discrete-candidates}}

\subsection{Dynamic programming algorithm and proof of correctness.}

For the algorithm, we will assume throughout that $s_{i-1} < s_i < s_{i+1}$, and that $\minu \ge \maxd$. If either of these is violated, the corresponding table entry is $0$.

\begin{algorithm}
    \caption{Initialization($\points_c,m,T,s_{m+1},s_0$)}\label{algo:dp-init}
    \begin{algorithmic}[1]
        \For{\text{each~}$s_{m-1}<s_m\in\points_c$, $\minu\ge\maxd\in\mcV$}
            \State $\strat' = \{s_{m-1},s_m\}$
            \If{\{$\util(s_m,\strat')\ge \intu(s_{m-1},s_{m+1})$,\newline \hphantom{0.25cm}$\util(s_m,\strat')\ge\minu$,\newline \hphantom{0.25cm}$\maxd\ge\{\intu(s_{m-1},s_m),\intu(s_m,s_{m+1})\}$\}}
                \State $T(m,s_{m-1},s_m,s_{m+1},\minu,\maxd)=1$
                
            \EndIf
        \EndFor
    \end{algorithmic}
\end{algorithm}

\begin{algorithm}[H]
    \caption{Discrete-PNE($\points_c,m$)}\label{algo:dp}
    \begin{algorithmic}[1]
        \State $s_{m+1} = +\infty$, $s_0 = -\infty$
        \State Initialize a table $T$ to 0 of dimension $m\times (n+2)\times (n+2)\times (n+2)\times 4n^3\times 4n^3$
        \State Initialization($\points_c,m,T,s_{m+1},s_0$)
        \For{each $i$ from $m-1$ to $1$}\Comment{Table entry $T(i,s_{i-1},s_i,s_{i+1},\minu,\maxd)$}
            \For{each $s_{i-1}<s_i<s_{i+1}\in\points_c,\minu\ge\maxd\in\mcV$}
                \State $\strat' = \{s_{i-1},s_i,s_{i+1}\}$
                \If{\{$\util(s_i,\strat')\ge \intu(s_{i-1},s_{i+1})$,\newline \hphantom{1cm}\hspace{1.2cm}$\util(s_i,\strat')\ge\minu$,\newline \hphantom{1cm}\hspace{1.2cm}$\maxd\ge\{\intu(s_{i-1},s_i),\intu(s_i,s_{i+1})\}$\}}
                   \For{each $s_{i+2}\in\points_c,\minu'\ge\maxd'$}
                        \If{\{$s_{i+2}>s_{i+1}$,\newline \hphantom{1cm}\hspace{2.4cm}$T(i+1,s_i,s_{i+1},s_{i+2},\minu',\maxd')=1$,\newline \hphantom{1cm}\hspace{2.4cm}$\minu'\ge\minu\ge\maxd\ge\maxd'$\}}
                            \State $T(i,s_{i-1},s_i,s_{i+1},\minu,\maxd)=1$
                            \State Break
                        \EndIf
                    \EndFor
                \EndIf
            \EndFor
        \EndFor
    \end{algorithmic}
\end{algorithm}

\begin{restatable}{lemma}{discretetablefilling}
\label{lem:discrete-table-filling}
    For $i = 1, \ldots, m$, $T(i,s_{i-1},s_i,s_{i+1},\minu,\maxd) = 1$ iff $\minu\ge\maxd$, and there exist positions $s_{i+2} <  \ldots < s_m \in \points_c$ so that for $\strat' = \{s_{i-1}, s_i, \ldots, s_m\}$ and $\points_c' = \{p \in \points_c :p \ge s_{i-1}\}$,
\begin{enumerate}[label=(\alph*)]
    \item \label{property:interval} for all $j \ge i$, $\util(s_j, \strat') = \intu(s_{j-1}, s_{j+1})$ (thus $s_j$ is a best response in the interval $(s_{j-1}, s_{j+1})$ given the positions of candidates $j-1$ and $j+1$),
    \item \label{property:minu} the minimum utility of any candidate $j \ge i$ is at least $\minu$, and
    \item \label{property:maxd} Given the strategy profile $\strat'$, the maximum utility obtainable by a new candidate in $\points_c'\setminus\strat'$ is at most $\maxd$.
\end{enumerate}
\end{restatable}

\begin{proof}
We first show that if $T(i,s_{i-1},s_i,s_{i+1},\minu,\maxd) = 1$, then the given conditions must be satisfied. We prove this by using induction on $i$ from $m$ to 1.

If $i=m$, then $\points_c' = \{p \in \points_c: p \ge s_{m-1}\}$ and $\strat' = \{s_{m-1},s_m\}$. The if condition at line 3 of Algorithm \ref{algo:dp-init} captures all the constraints from \cref{property:interval,property:minu,property:maxd}.\footnote{Recall that $s_{m+1}=+\infty$. So we ignore the utility of the dummy candidate $m+1$.} 


If $i<m$, then $\points_c' = \{p \in \points_c: p \ge s_{i-1}\}$, and by the if condition in line 9 of Algorithm \ref{algo:dp} there exist $s_{i+2} > s_{i+1}$, and $\minu'$, $\maxd'$ so that $\minu' \ge \minu \ge \maxd \ge \maxd'$ and $T(i+1, s_i, s_{i+1}, s_{i+2}, \minu', \maxd') = 1$. By the inductive hypothesis,~\cref{property:minu,property:maxd} hold for $T(i+1, s_i, s_{i+1}, s_{i+2}, \minu', \maxd')$. That is, there exist $s_{i+3}$, $\ldots$, $s_m$ so that each candidate $j \ge i+1$ is in a best position in the interval $(s_{j-1}, s_{j+1})$, each candidate $j \ge i+1$ gets at least $\minu'$ utility, and no new candidate can get more than $\maxd'$ utility at any position after $s_i$.

The if condition of line 7 of Algorithm \ref{algo:dp} first checks that candidate $i$ is in a best position in the interval $(s_{i-1},s_{i+1})$, hence this now holds for all $j \ge i$, satisfying~\cref{property:interval}. Secondly, candidate $i$'s utility is at least $\minu$, hence each candidate $j \ge i$ has utility at least $\minu$ (since $\minu' \ge \minu$). Thirdly, it checks that $\maxd \ge \max \{ \intu(s_{i-1},s_i), \intu(s_i, s_{i+1})\}$. Hence any new candidate that takes a position in $\points_c'$ gets utility at most $\maxd$ (since $\maxd' \le \maxd)$, satisfying~\cref{property:maxd}.

We next show that if the conditions of Lemma~\ref{lem:discrete-table-filling} are satisfied, then $T(i,s_{i-1},s_i,s_{i+1},\minu,\allowbreak\maxd) = 1$. We prove this by using induction on $i$ from $m$ to $1$. Let $s_{i-1}$, $s_i$,$\ldots$, $s_m$ be candidate positions satisfying~\cref{property:interval,property:minu,property:maxd}.

For the base case $i=m$. Then the if condition of line 3 of Algorithm \ref{algo:dp-init} is satisfied which implies $T(m,s_{m-1},s_m,s_{m+1},\minu,\maxd)=1$.

By the induction hypothesis, $T(i+1,s_i,s_{i+1},s_{i+2},\minu,\maxd)=1$, since clearly the positions $s_{i}, \ldots, s_m$ and $\minu,\maxd$ satisfy~\cref{property:interval,property:minu,property:maxd} for $i+1$. The conditions in the lemma also stipulate that candidate $i$ is in a best position in $(s_{i-1}, s_{i+1})$, $i$'s utility is at least $\minu$, and a new candidate that takes a position in either $(s_{i-1},s_i)$ or $(s_i, s_{i+1})$ gets utility at most $\maxd$. Therefore the if condition at line 10 of Algorithm \ref{algo:dp} is satisfied which implies $T(i,s_{i-1},s_i,s_{i+1},\minu,\maxd)=1$.    
\end{proof}

\thmdp*

\begin{proof}
It follows from~\Cref{lem:discrete-table-filling} and~\Cref{prop:equilibrium-characterization} that there exist $s_0 = -\infty$, $s_1$, $s_2$, $\minu \ge \maxd$ so that $T(1, s_1, s_2, s_3, \minu, \allowbreak\maxd) = 1$ if and only if the instance has an equilibrium. Note that there are at most $m\times (n+2)^3\times |\mcV|^2$ table entries. For each cell in the table—specified by $ i, s_{i-1}, s_i, s_{i+1}, \minu, \maxd $—Algorithm \ref{algo:dp} determines values $ s_{i+2}, \minu', \maxd' $ such that $ s_{i+2} $ is a valid candidate location, $ \minu' \geq \minu \geq \maxd \geq \maxd' $, and $ T(i+1, s_i, s_{i+1}, s_{i+2}, \minu', \allowbreak\maxd') = 1 $. Therefore Algorithm \ref{algo:dp}, decides and finds a PNE in $\mathcal{O}(m\times (n+2)^3\times |\mcV|^2\times n\times |\mcV|^2)$.
\end{proof}

\subsection{Proof of Theorem~\ref{thm:dpapprox}}

We first modify~\Cref{prop:equilibrium-characterization} appropriately.

\begin{prop}
\label{prop:equilibrium-characterization-approx}
    Strategy profile $\strat = \{s_1, \ldots, s_m\}$ where $s_1 < s_2 < \ldots < s_m$ is an $\epsilon$-equilibrium if and only if:\footnote{Recall that we assume there are dummy candidates positioned at $s_0 = -\infty$ and $s_{m+1} = +\infty$.} 
    \begin{enumerate}[label=(\roman*)]
        \item for all $j \in [m]$, $\util(s_j,\strat) \ge \intu(s_{j-1},s_{j+1}) - \epsilon$, and
        \item  $\min_{j \in [m]} u(s_j, \strat) \ge \max_{k \in \{0, \ldots, m\}} \intu(s_k, s_{k+1}) - \epsilon$.
    \end{enumerate}
\end{prop}

\begin{proof} For the if direction, for the sake of contradiction, assume that $\strat = \{s_1, \ldots, s_m\}$ is not an $\epsilon$-equilibrium. Therefore, there exists a candidate $i$ who can improve her utility by more than $\epsilon$ at another location $s_i' \in \points_c\setminus\strat$. If $s_{i-1} < s_i' < s_{i+1}$, this contradicts the first condition of the Proposition, since $\util(s_j,\strat) \ge \intu(s_{i-1}, s_{i+1}) - \epsilon$. Now, suppose $s_i' \notin [s_{i-1}, s_{i+1}]$. According to the second condition of the proposition, $\util(s_i, \strat) \geq \max_{k \in \{0, \ldots, m\}} \intu(s_k, s_{k+1}) - \epsilon$, which contradicts the assumption that candidate $i$ can strictly increase her utility by more than $\epsilon$ by moving to position $s_i'$.

        For the other direction, let $\strat$ be an $\epsilon$-equilibrium and consider an arbitrary candidate $i \in [m]$. Then, $s_i$ must be an $\epsilon$-best response for candidate $i$, which means $\util(s_i,\strat) \ge \intu(s_{i-1}, s_{i+1}) - \epsilon$. Now, let $i_{\min} = \arg \min_{j \in [m]} u(s_j, \strat)$ be a candidate with the minimum utility. Since $\strat = \{s_1, \ldots, s_m\}$ is an $\epsilon$-equilibrium, $s_{i_{\min}}$ is also an $\epsilon$-best response for candidate $i_{\min}$. Thus, $\util(s_{i_{\min}}, \strat) \geq \max_{k \in \{0, \ldots, m\}} \intu(s_k, s_{k+1}) - \epsilon$, which is the second condition of the proposition.
\end{proof}

Now the table entry $T(i,s_{i-1},s_i,s_{i+1},\minu, \maxd)$ is set to 1 if the following conditions are satisfied, with $\strat' = \{s_{i-1}, s_i, s_{i+1}\}$:

\begin{enumerate}
    \item[C1':] \label{cond:c1prime} $\util(s_i, \strat') \ge \intu(s_{i-1},s_{i+1}) - \epsilon$, $\util(s_i, \strat')$ $ \ge \minu$ $\ge \maxd - \epsilon$,  \\ $\maxd \ge \max \{\intu(s_{i-1},s_i), \intu(s_i, s_{i+1})\}$, and
    \item[C2':] \label{cond:c2prime} there exist $s_{i+2}  \in \points_C$ and $\minu', \maxd' \in \mcV$ so that $s_{i+2}> s_{i+1}$, $\minu' \ge \minu \ge \maxd - \epsilon \ge\maxd' - \epsilon$, and $T(i+1,s_i, s_{i+1}, s_{i+2}$ $, \minu', \maxd') = 1$.
\end{enumerate}

If $i = m$, the second condition is not checked. If there exist $s_1$, $s_2$, $s_3$, $\minu$, $\maxd$ so that $T(1, s_1, s_2, s_3, \minu, \maxd) =1$, the algorithm returns that the instance has an equilibrium, else it returns that there is no equilibrium.

\begin{lemma}
\label{lem:discrete-table-filling-approx}
    For $i = 1, \ldots, m$, $T(i,s_{i-1},s_i,s_{i+1},\minu,\maxd) = 1$ iff $\minu\ge\maxd - \epsilon$, and there exist positions $s_{i+2} <  \ldots < s_m \in \points_c$ so that for $\strat' = \{s_{i-1}, s_i, \ldots, s_m\}$ and $\points_c' = \{p \in \points_c :p \ge s_{i-1}\}$,
\begin{enumerate}[label=(\alph*)]
    \item \label{property:interval-approx} for all $j \ge i$, $\util(s_j, \strat') \ge \intu(s_{j-1}, s_{j+1}) - \epsilon$,
    \item \label{property:minu-approx} the minimum utility of any candidate $j \ge i$ is at least $\minu$, and
    \item \label{property:maxd-approx} given the strategy profile $\strat'$, the maximum utility obtainable by a new candidate in $\points_c'\setminus\strat'$ is at most $\maxd$.
\end{enumerate}
\end{lemma}

\begin{proof}
We first show that if $T(i,s_{i-1},s_i,$ $s_{i+1},\minu,\maxd) = 1$, then the given conditions must be satisfied. We prove this by using induction on $i$ from $m$ to 1.

If $i=m$, then $\points_c' = \{p \in \points_c: p \ge s_{m-1}\}$ and $\strat' = \{s_{m-1},s_m\}$. If $T(m,s_{m-1},s_m,$ $s_{m+1},\minu,\allowbreak\maxd) = 1$, then from condition~\ref{cond:c1prime}, $\util(s_m, \strat') \ge \intu(s_{m-1}, s_{m+1}) - \epsilon$, satisfying the first condition in the lemma. Further from the condition, $\util(s_m, \strat') \ge \minu$, and $\maxd \ge \max \{\intu(s_{m-1},s_m),$ $ \intu(s_m, s_{m+1})\}$, satisfying the other conditions as well.

If $i<m$, then $\points_c' = \{p \in \points_c: p \ge s_{i-1}\}$, and $\strat' = \{s_{i-1}, s_i, \ldots, s_m\}$. If $T(i,s_{i-1},s_i,$ $s_{i+1},\minu,\maxd) = 1$, then from condition~\ref{cond:c1prime}, $\util(s_i, \strat') \ge \intu(s_{i-1}, s_{i+1}) - \epsilon$, satisfying the first condition in the lemma. Further from the condition, $\util(s_i, \strat') \ge \minu$, and $\maxd \ge \max \{\intu(s_{i-1},s_i), $ $\intu(s_i, s_{i+1})\}$. 

By the inductive hypothesis,~\cref{property:minu-approx,property:maxd-approx} hold for $T(i+1, s_i, s_{i+1},$ $ s_{i+2}, \minu', \maxd')$. That is, there exist $s_{i+3}$, $\ldots$, $s_m$ so that each candidate $j \ge i+1$ is in an $\epsilon$-best position in the interval $(s_{j-1}, s_{j+1})$, each candidate $j \ge i+1$ gets at least $\minu'$ utility, and no new candidate can get more than $\maxd'$ utility at any position after $s_i$. Condition~\ref{cond:c2prime} then checks if $\minu' \ge \minu \ge \maxd - \epsilon \ge\maxd' - \epsilon$. It follows that for all $j \ge i$, $\util(s_j, \strat') \ge \intu(s_{j-1}, s_{j+1}) - \epsilon$, satisfying the first condition of the lemma. Further, $\util(s_j, \strat') \ge \minu$ (satisfying the second condition), and $\maxd$ is at least the utility to be gained by a new candidate in $\points_c' \setminus \strat'$ satisfying the third condition.

The proof of the opposite direction, to show that if the conditions of Lemma~\ref{lem:discrete-table-filling-approx} are satisfied, then $T(i,s_{i-1},s_i,s_{i+1},\minu,\maxd) = 1$, is very similar and is omitted.
\end{proof}

\thmdpapprox*

\begin{proof}
The proof is similar to the proof for~\Cref{thm:dp}. It follows from~\Cref{lem:discrete-table-filling-approx} and~\Cref{prop:equilibrium-characterization-approx} that there exist $s_0=-\infty$, $s_1$, $s_2$, $\minu \ge \maxd-\epsilon$ so that $T(1, s_0, s_1, s_2, \minu, \maxd) = 1$ if and only if the instance has an equilibrium. The runtime analysis is the same as earlier.  Therefore the modified algorithm decides and finds an $\epsilon$-equilibrium in time $\mathcal{O}(m\times (n+2)^3\times |\mcV|^2\times n\times |\mcV|^2)$.
\end{proof}


\section{Proofs from Section \ref{sec:continuous-candidates}}

We now prove~\Cref{thm:lowcomplexityalgo}. For this, we need the following lemma, which shows that even if $\points_c$ is a continuous interval, a candidate's best response can be computed in polynomial time. 

\begin{restatable}[]{lemma}{bestu}
\label{lm:best-u}
    Given an instance $\mcI=(\points_v,m,F,R)$ and a strategy profile $\strat$, in any interval between two candidates $(s_i, s_{i+1})$, the best utility $\intu(s_i, s_{i+1})$ and a best response $\intr(s_i, s_{i+1})$ can be computed in time $\poly(|\mcI|)$.
\end{restatable}

\begin{proof}
    Let $\points_v' = \points_v \cap [s_i, s_{i+1}]$ be the set of voter locations in $[s_i, s_{i+1}]$. Define $X = \{x_1, x_2, \ldots, x_t\}$ as the sorted set $(B(s_i, \points_v') \cup B(s_{i+1}, \points_v') \cup \{s_i,s_{i+1}\}) \cap [s_i,s_{i+1}]$, where $x_1 = s_i$ and $x_t = s_{i+1}$. Now from Case I of the proof of~\Cref{lm:1}, it follows that the utility obtainable by a new candidate is constant within each interval $(x_j, x_{j+1})$ for $j \in [t-1]$. Hence, to compute the best utility, it suffices to check all points in the set $\{x_1, \mu(x_1,x_2), x_2, \ldots, x_{k-1}, \mu(x_{k-1},x_k), x_k\}$, which is of size $\mathcal{O}(|\points_v'|)$. \end{proof}

\lowcomplexityalgo*

\begin{proof}
    The proof follows the dynamic program that constructs the table with entries $T(i,s_{i-1},\allowbreak s_i,s_{i+1},\minu,\maxd)$. Then as per the algorithm, if there exist $s_1$, $s_2 \in \hat{\points}_c$, $\minu$, $\maxd$ so that $T(1, s_0 = \infty, s_1, s_2, \minu, \maxd) =1$, then the algorithm outputs that the instance indeed has an equilibrium. We first show that if there exist values $s_1$, $s_2\in \hat{\points}_c$, $\minu$, $\maxd$ such that  
\[
T(1, -\infty, s_1, s_2, \minu, \maxd) = 1,
\]
then an equilibrium indeed exists in $\points_c$. From~\Cref{lem:discrete-table-filling}, it follows that there exist locations $s_3, \ldots, s_m \in \hat{\points}_c$, with $s_2<s_3 < \ldots < s_m$, such that the strategy profile $\strat = \{s_1, \ldots, s_m\}$ satisfies:
\begin{enumerate}
    \item for all $j \ge 1$, $\util(s_j, \strat) = \intu(s_{j-1}, s_{j+1})$,
    \item the utility of every candidate $j \ge 1$ is at least $\minu$,
    \item for all $j \in [m+1]$, $\maxd \ge \intu(s_{j-1}, s_j)$, and
    \item $\minu \ge \maxd$.
\end{enumerate}

Since $\intu(s_j,s_{j+1})$ is computed by taking maximum over $(s_j, s_{j+1}) \cap \points_c$ (\Cref{lm:best-u}), it follows from~\Cref{prop:equilibrium-characterization} that $\strat$ constitutes an equilibrium in $\points_c$.

Now suppose an equilibrium exists in $\points_c$. We show that the algorithm must output an equilibrium. In this case, by~\Cref{thm:lowcomplexity}, there also exists an equilibrium where each candidate is at a location in $\hat{\points}_c$, denoted by $\strat = \{s_1, \ldots, s_m\}$. It is easy to see that for all $i \in [m]$, and for $\minu = \min_i \util(s_i,\strat)$, $\maxd=\max_{s \not \in \strat} \util(s, \strat \cup\{s\})$, the table entry $T(i,s_{i-1},s_i, s_{i+1},\minu,\maxd) = 1$. The theorem then follows.
\end{proof}

\subsection{Figure \ref{fig:shifting1} is an equilibrium.}

Figure \ref{fig:shifting1} represents a PNE with 5 candidates and 24 voters. The red circles indicate the voters' locations, with the numbers below each circle representing the number of voters at that location. The blue numbers correspond to the lengths of the respective line segments. 

The figure shows an equilibrium strategy profile $\strat = \{0,2,9,16,18\}$ with the candidates shown as black triangles. We verify now that this is indeed an equilibrium, and in fact is almost unique --- at any equilibrium

\begin{enumerate}
    \item there must be candidates at positions $a$, $b$, $j$, $k$,
    \item there must be a candidate in the interval $[d,h]$.
\end{enumerate}

Firstly, to verify that the given profile is an equilibrium, consider candidates 1, 2, 4, and 5. Note that the voters at $a$, $b$, $j$, and $k$ have these candidates respectively at distance 0, so these will not vote for another candidate unless the candidate at their position deviates. Thus candidates 1, 2, 4, and 5 can at most obtain the votes of the voters at $d$ and $h$. But to get these votes, it is clear that they must give up their existing votes, which would reduce their vote share.

Given the positions of candidates 1, 2, 4, and 5, candidate 3 is then maximizing her vote share, and hence this is an equilibrium.

Now to see that the equilibrium is almost unique, note that the instance has 24 voters and 5 candidates, hence in any strategy profile, some candidate must receive less than 5 votes. If any of the positions $a$, $b$, $j$, or $k$ are unoccupied, this candidate can deviate to the unoccupied position and increase her vote share. Hence in any equilibrium, these positions must be occupied by a candidate.  

Now let us consider what happens as we shift candidate 3 to the right. We will show that if $s_3$ crosses the dashed line at $g$, the strategy profile is no longer an equilibrium. This is because as mentioned earlier in the explanation, if candidate 3 shifts to a position $g + \delta$, her distance from the voters at $d$ --- that were voting for her in $\strat$ --- is $3+\delta$. But this is not an equilibrium, since now candidate 2 can increase her vote share. Candidate 2 can shift to the position $c- \delta/2$. Then her distance to $b$ is $2 - \delta/2$, and her distance to $d$ is $3 + \delta/2$. She is then the closest candidate to both these positions, and her vote share increases from 5 to 7.

\subsection{Candidates may not be at mid-points at equilibrium}
\label{sec:B-pne}

We give an instance with a unique equilibrium $\strat$ and a candidate $i$ such that $s_i \not \in \points_v \cup \{\mean(x,y):x, y \in \points_v\}$. Thus, restricting attention to points in $\points_v$ and in the middle of points in $\points_v$ is not sufficient to obtain an equilibrium.

\begin{figure}[ht]
 \centering
    \begin{tikzpicture}[line cap=round,line join=round,>=triangle 45,x=1cm,y=1cm,scale=0.6]

\draw [line width=1pt] (0,0) -- (20,0);

\draw [line width=1pt,color=blue,dashed] (4,-0.7) -- (4,0.7);
\draw [line width=1pt,color=blue,dashed] (8,-0.7) -- (8,0.7);
\draw [line width=1pt,color=blue, dashed] (14,-0.7) -- (14,0.7);
\begin{footnotesize}
\draw [fill=red] (0,0) circle (6pt);
\draw [fill=red] (2,0) circle (6pt);
\draw [fill=red] (6,0) circle (6pt);
\draw [fill=red] (11,0) circle (6pt);
\draw [fill=red] (17,0) circle (6pt);
\draw [fill=red] (20,0) circle (6pt);

\node[isosceles triangle,
    draw,
    rotate=90,
    fill=black,
    minimum size =0.05cm, scale=0.5] (T1)at (0,1.3){};

\node[isosceles triangle,
    draw,
    rotate=90,
    fill=black,
    minimum size =0.05cm, scale=0.5] (T2)at (2,1.3){};
\node[isosceles triangle,
    draw,
    rotate=90,
    fill=black,
    minimum size =0.05cm, scale=0.5] (T3)at (8,1.3){};
\node[isosceles triangle,
    draw,
    rotate=90,
    fill=black,
    minimum size =0.05cm, scale=0.5] (T4)at (17,1.3){};
\node[isosceles triangle,
    draw,
    rotate=90,
    fill=black,
    minimum size =0.05cm, scale=0.5] (T5)at (20,1.3){};

\end{footnotesize}
\draw[color=blue] (1,0.5) node {$2$};
\draw[color=blue] (3,0.5) node {$2$};
\draw[color=blue] (5,0.5) node {$2$};
\draw[color=blue] (7,0.5) node {$2$};
\draw[color=blue] (9.5,0.5) node {$3$};
\draw[color=blue] (12.5,0.5) node {$3$};
\draw[color=blue] (15.5,0.5) node {$3$};
\draw[color=blue] (18.5,0.5) node {$3$};

\draw[color=black] (0,2) node {$s_1$};
\draw[color=black] (2,2) node {$s_2$};
\draw[color=black] (8,2) node {$s_3$};
\draw[color=black] (17,2) node {$s_4$};
\draw[color=black] (20,2) node {$s_5$};

\foreach \x in {0,2,17,20} {
        \draw (\x,-0.8) node[circle, inner sep=1pt, draw, color=red, fill=white] {5}; 
    }

\foreach \x in {6,11} {
        \draw (\x,-0.8) node[circle, inner sep=1pt, draw, color=red, fill=white] {2}; 
    }

\foreach \x/\y in {0/a,2/b,4/c,6/d,8/e,11/f,14/g,17/h,20/i  } {
\draw[color=black] (\x,-1.8) node {$\y$};
}

\end{tikzpicture}
    \caption{An equilibrium with five candidates and 24 voters. The red circles indicate the voters' locations, with the numbers below each circle representing the number of voters at that location. The blue numbers correspond to the lengths of the respective line segments}
    \Description{An equilibrium with five candidates and 24 voters. The red circles indicate the voters' locations, with the numbers below each circle representing the number of voters at that location. The blue numbers correspond to the lengths of the respective line segments}
    \label{fig:B-pne}
\end{figure}

\begin{example}\label{example:1}
Figure~\ref{fig:B-pne} presents equilibrium in an instance with 5 candidates and 24 voters. Note that $s_3\notin B(s_2,\points_v)\cup B(s_4,\points_v)\cup\points$. If we attempt to move candidate 3 to any other location, whether to the left or right, it disrupts the equilibrium. Additionally, note that $s_3$ is not positioned at the midpoint of any pair of voters' locations.

Now we show that this is the only equilibrium possible in the instance of Figure \ref{fig:B-pne}. Note that there are 24 voters and 5 candidates, therefore in any strategy profile one of the candidates gets less than 5 votes. Therefore, in equilibrium, there are always some candidates on $a$, $b$, $h$, and $i$. Now if the fifth candidate chooses $e-\epsilon$ (with $\epsilon>0$) as her strategy, then the candidate on $h$ can increase her votes by deviating to $g+(\epsilon/2)$. Similarly, if the fifth candidate chooses $8+\epsilon$ ($\epsilon>0$) as her strategy, the candidate on $b$ can increase her votes by deviating to $c-(\epsilon/2)$. Therefore the equilibrium depicted in Figure \ref{fig:B-pne} is the only equilibrium possible.
\end{example}

\subsection{Proof of Lemma \ref{lm:1}.}

We define the left and right utility of a candidate and a new candidate as follows.

\begin{align*}
    \util^L(s,\strat) &= \displaystyle \sum_{p \in \points \text{ and } p<s} \frac{w(p)}{|\near(p,\strat)|} \times \mone_{s \in \near(p,\strat)} \,\\
    \util^R(s,\strat) &= \displaystyle \sum_{p \in \points \text{ and } p>s} \frac{w(p)}{|\near(p,\strat)|} \times \mone_{s \in \near(p,\strat)} \,\\
    \utilout^{,L}(s',\strat) &= \util^L(s', \strat \cup \{s'\})\\
    \utilout^{,R}(s',\strat) &= \util^R(s', \strat \cup \{s'\})
\end{align*}
Note that $\util(s,\strat)=\util^L(s,\strat)+\util^R(s,\strat)+w(s)$ and $\util(s,\strat)=\utilout(s,\strat\setminus\{s\})$.

\lmone*

\begin{proof}
Throughout the proof, we denote the strategy of candidate $j\in[m]$ in $\strat'$ as $s_j'$. Thus, $s_i'=x$ and $s_j'=s_j$, for $j\neq i$. Clearly, $\forall j'\in\{1,...,i-2,i+2,...,m\}$, $\util(s_{j'},\strat)=\util(s_{j'}',\strat')$.
    \paragraph*{Case I:} Suppose $j=i-1$. Clearly, $\util^L(s_{i-1},\strat)=\util^L(s_{i-1}',\strat')$. We know that, $B(s_{i-1},\points_v)\cap [s_i,x]=\emptyset$. Therefore, $\util^R(s_{i-1},\strat)=\util^R(s_{i-1}',\strat')$. We conclude that $\util(s_{i-1},\strat)=\util(s_{i-1}',\strat')$. 
    \paragraph*{Case II:} Suppose $j=i+1$. As $x>s_i$, the utility of candidate $i+1$ can only decrease in $\strat'$. From Case I, we know that the utility of the candidate $i-1$ does not change in $\strat'$. Therefore, $\util(s_{i+1},\strat)>\util(s_{i+1}',\strat')$ implies $\util(s_{i},\strat)<\util(s_{i}',\strat')$ contradicting the fact that $\strat$ is a PNE.

    As for all candidates $j'\neq i$, $\util(s_{j'},\strat)=\util(s_{j'}',\strat')$, it follows that $\util(s_i,\strat)=\util(s_i',\strat')$.
\end{proof}

\subsection{Proof of Lemma \ref{lm:3}.}

We first show the following.

\begin{lemma}\label{lm:2}
     For $m>2$, let $\strat$ be a PNE. Consider a candidate $i\in[m]$. Let $x\in(s_i,s_{i+1})$ and $\strat' = \{s_1,...,s_{i-1},x,s_{i+1},...,s_m\}$ be such that $\util(s,\strat)=\util(s',\strat')$. Then in $\strat'$, no candidate can improve their utility by choosing a strategy $y\in(x,1]$.
\end{lemma}
\begin{proof}
    Observe that candidate $i$ has no better response in the interval $(x,1]$, since $\strat$ is a PNE. We divide the proof into two cases.
    \paragraph*{Case I:} Consider a candidate $j\neq i+1$. Deviating to location $y\in(s_{i+1},1]$ cannot increase $j$'s utility. Note that for all $y\in(x,s_{i+1})$, $\utilout^{,L}(y,\strat')\le\utilout^{,L}(y,\strat)$, and $\utilout^{,R}(y,\strat')=\utilout^{,R}(y,\strat)$. Therefore, $\utilout(y,\strat')\le\utilout(y,\strat)$. We conclude that candidate $j\neq i+1$ has no better deviation in the interval $(x,s_{i+1})$.
    \paragraph*{Case II:} Consider the candidate $j=i+1$. Deviating to location $y\in(s_{i+2},1]$ cannot increase $(i+1)$'s utility. Note that for all $y\in(x,s_{i+2})$, $\utilout^{,L}(y,\strat'\setminus\{s_j\})\le\utilout^{,L}(y,\strat\setminus\{s_j\})$, and $\utilout^{,R}(y,\strat'\setminus\{s_j\})=\utilout^{,R}(y,\strat\setminus\{s_j\})$. Therefore, $\utilout(y,\strat'\setminus\{s_j\})\le\utilout(y,\strat\setminus\{s_j\})$. We conclude that candidate $i+1$ has no better deviation in the interval $(x,s_{i+2})$.
    
\end{proof}

\shifting*

\begin{proof}
    Throughout the proof, we denote the strategy of candidate $j\in[m]$ in $\strat'$ as $s_j'$. Thus, $s_i'=x$ and $s_j'=s_j$, for all $j\neq i$. For the strategy profile $\strat'$, it is easy to see that there is no better deviation for any candidate in $[0,s_{i-2}')\cup (s_{i+1}',1]$. Also from Lemma \ref{lm:1} and \ref{lm:2}, we conclude that no candidate has a better deviation in the interval $(x,s_{i+1}')$. For the sake of contradiction suppose candidate $k$ has a better deviation $y \in(s_{i-2}',s_i']$. It is easy to see that $k\neq i$ since $i$ could have deviated to $y$ in $\strat$ also. We divide our proof into two cases. Let $\lambda = x-s_i$.
\paragraph*{Case I:}Suppose $k\ne i-1$. Then, $y\in(s_{i-1}',s_i')$ otherwise $k$ could have deviated to $y$ in $\strat$ also. Consider the set, $Z = \{q\in B(s_i,\points_v) ~|~y\le q\le(y+\lambda)\}$. Therefore, $Z$ represents the set of reflections of $s_i$ with respect to $\points_v$ that are initially to the right or at $y$, but after candidate $i$ moves to $x$, these reflections shift to the left or at $y$. This is shown in Figure \ref{fig:z-max}.
        \begin{clm}\label{clm:2}
            $Z$ is non-empty.
        \end{clm}
        \begin{proof}
            For the sake of contradiction suppose $Z$ is empty. Therefore, $\utilout^{,R}(y,\strat)=\utilout^{,R}(y,\strat')$. Moreover, candidate $i-1$ does not move in $\strat'$, so $\utilout^{,L}(y,\strat)=\utilout^{,L}(y,\strat')$. Therefore, $\utilout(y,\strat) = \utilout(y,\strat')$. However, we assumed that $y$ is a better response for some candidate $k$ in $\strat'$, while it is not a better response for any candidate in $\strat$ since $\strat$ is a PNE, leading to a contradiction.
        \end{proof}
        It is important to note that we did not utilize the set $B^2()$ in our proof. However, we will be using it for the next claim. Suppose $z_{\max}=\max(Z)$.
       \begin{clm}\label{clm:3}
            There is no $b\in[y,z_{\max}]$ such that $b\in B(s_{i-1},\points_v)$.
        \end{clm}
        \begin{proof}
            We know that there exists a point $p\in \points_v$ such that $f_r(z_{\max},p)=s_i$. For the sake of contradiction suppose $\exists b\in[y,z_{\max}]$ such that $b\in B(s_{i-1},\points_v)$. From the definition of $z_{\max}$, we conclude that, $z_{\max}\ge y\ge (z_{\max}-\lambda)$. Note that $f_r(z_{\max}-\lambda,p)=x$. Therefore there exists a $b'\in[s_i,x]$ such that $f_r(b,p)=b'$ which implies $b'\in B^2(s_{i-1},\points_v)$. This contradicts the definition of $x$.
        \end{proof} 
        Now we try to find a point $y'$ such that $\utilout(y,\strat')\le\utilout(y',\strat)$ which contradicts that $\strat$ is a PNE. Suppose $b=\min\left(\{q\in B(s_{i-1},\points_v)\mid q\ge y\}\right)$. Note that $b$ may not exist or may be greater than $s_i$. That's why suppose $b'=\min(s_i,b)$. Observe that $b'>z_{\max}$ (Claim \ref{clm:3}). Let $y'=((z_{\max}+b')/2)$. Observe the following inequalities.
        \begin{align}
            \utilout^{,R}(y',\strat)&\ge\utilout^{,R}(y,\strat')-\sum_{p\in\points_v\cap(y,y']}F(p)\label{ieq:1}\\
            \utilout^{,L}(y',\strat)&\ge\utilout^{,L}(y,\strat')+\sum_{p\in\points_v\cap[y,y')}F(p)\label{ieq:2}
        \end{align}
        Therefore we have,
        \begin{align*}
            \utilout(y',\strat)&=\utilout^{,L}(y',\strat)+\utilout^{,R}(y',\strat)+F(y')\\
            &\ge\utilout^{,L}(y,\strat')+\sum_{p\in\points_v\cap{[y,y')}}F(p)+\utilout^{,R}(y,\strat')-\sum_{p\in\points_v\cap{(y,y']}}F(p)+F(y')~~~~\tag{\text{From (\ref{ieq:1}}) and (\ref{ieq:2})}\\
            &= \utilout^{,L}(y,\strat')+\utilout^{,R}(y,\strat')+F(y)\\
            &=\utilout(y,\strat')
        \end{align*}
        Hence we have $\utilout(y',\strat)\ge\utilout(y,\strat')$ which contradicts that $\strat$ is a PNE.

\begin{figure}
    \centering
    \begin{tikzpicture}[line cap=round,line join=round,>=triangle 45,x=1cm,y=1cm,scale=0.55]
    \clip(-1,-3) rectangle (23,3);

\draw [line width=1pt] (0,0) -- (21.5,0);
\draw[line width=1pt](21.5,-0.2)--(21.5,0.2);
\draw [fill=xdxdff] (8.5,0) circle (6pt);
\draw [fill=xdxdff] (12,0) circle (6pt);
\draw [fill=xdxdff] (16,0) circle (6pt);

\node[mark size=4pt,color=black] at (0,0) {\pgfuseplotmark{triangle*}};
\node[mark size=4pt,color=black] at (18,0) {\pgfuseplotmark{triangle*}};
\node[mark size=4pt,color=black] at (11,0) {\pgfuseplotmark{square*}};
\draw [line width=1pt, color = blue] (17,-1) -- (17,1);
\draw [line width=1pt, color = red] (14,-1) -- (14,1);
\draw [line width=1pt, color = black] (15.5,-1) -- (15.5,1);
\draw [line width=0.3pt, color = black,<->] (18,0.5)--(21.5,0.5);

\draw [dashed, line width=1pt, color = red] (10.5,-1) -- (10.5,1);


\draw[color=xdxdff] (8.5,-0.8) node {$p_{k-1}$};
\draw[color=xdxdff] (12,-0.8) node {$p_{k}$};
\draw[color=xdxdff] (16.3,-0.8) node {$p_{k+1}$};
\draw[color=black] (0,-0.8) node {$s_{i-1}$};
\draw[color=black] (18,-0.8) node {$s_{i}$};
\draw[color=black] (15.5,-1.7) node {$y'$};
\draw[color=blue] (17,-1.7) node {$b'$};
\draw[color=red] (14.2,-1.7) node {$z_{\max}$};
\draw[color=red] (10.5,1.7) node {$(z_{\max}-\lambda)$};
\draw[color=black] (11.1,-0.8) node {$y$};
\draw[color=black] (21.5,-0.8) node {$x$};
\draw[color=black] (19.75,1) node {$\lambda$};

\begin{scriptsize}
    \draw[color=blue] (17,1.7) node {$f_r(s_{i-1},p_{k-1})$};
\draw[color=red] (14,1.7) node {$f_r(s_i,p_{k+1})$};

\end{scriptsize}
\end{tikzpicture}
    \caption{Case I of the proof of Lemma \ref{lm:3}}
    \Description{Case I of the proof of Lemma \ref{lm:3}}
    \label{fig:z-max}
\end{figure}

    \paragraph*{Case II:} Suppose $k=i-1$. Suppose $\strat_{-(i-1)} = \strat\setminus\{s_{i-1}\}$ and $\strat_{-(i-1)}' = \strat'\setminus\{s_{i-1}'\}$. Therefore $\util(s_{i-1},\strat)=\utilout(s_{i-1},\strat_{-(i-1)})$ and $\util(s_{i-1}',\strat')=\utilout(s_{i-1}',\strat_{-(i-1)}')$. From the assumption, $\utilout(y,\strat_{-(i-1)}')>\util(s_{i-1}',\strat')$. Consider the set, $Z = \{q\in B(s_i,\points_v) ~|~y\le q\le(y+\lambda)\}$. Similar to the proof of Claim \ref{clm:2}, we can show that $Z$ is non-empty. Let $z_{\max} =\max(Z)$.
       \begin{clm}\label{clm:4}
            There is no $b\in[y,z_{\max}]$ such that $b\in B(s_{i-2},\points_v)$.
        \end{clm}
        \begin{proof}
            We know that there exists a point $p\in \points_v$ such that $f_r(z_{\max},p)=s_i$. For the sake of contradiction suppose $\exists b\in[y,z_{\max}]$ such that $b\in B(s_{i-2},\points_v)$. From the definition of $z_{\max}$, we conclude that, $z_{\max}\ge y\ge (z_{\max}-\lambda)$. Note that $f_r(z_{\max}-\lambda,p)=x$. Therefore there exists a $b'\in[s_i,x]$ such that $f_r(b,p)=b'$ which implies $b'\in B^2(s_{i-2},\points_v)$. This contradicts the definition of $x$.
        \end{proof}

    Now we try to find a point $y'$ such that $\utilout(y,\strat'_{-(i-1)})\le\utilout(y',\strat_{-(i-1)})$ which contradicts that $\strat$ is a PNE. Suppose $b=\min\left(\{q\in B(s_{i-2},\points_v)\mid q\ge y\}\right)$. Note that $b$ may not exist or may be greater than $s_i$. That's why suppose $b'=\min(s_i,b)$. Observe that $b'>z_{\max}$. Let $y'=((z_{\max}+b')/2)$. Observe the following inequalities.
        \begin{align}
            \utilout^{,R}(y',\strat_{-(i-1)})&\ge\utilout^{,R}(y,\strat_{-(i-1)}')-\sum_{p\in\points_v\cap{(y,y']}}F(p)\label{ieq:3}\\
            \utilout^{,L}(y',\strat_{-(i-1)})&\ge\utilout^{,L}(y,\strat_{-(i-1)}')+\sum_{p\in\points_v\cap{[y,y')}}F(p)\label{ieq:4}
        \end{align}
        Therefore we have,
        \begin{align*}
            \utilout(y',\strat_{-(i-1)})&=\utilout^{,L}(y',\strat_{-(i-1)})+\utilout^{,R}(y',\strat_{-(i-1)})+F(y')\\
            &\ge\utilout^{,L}(y,\strat_{-(i-1)}')+\sum_{p\in\points_v\cap{[y,y')}}F(p)+ (\utilout^{,R}(y,\strat_{-(i-1)}')\\&\qquad\qquad\qquad\qquad-\sum_{p\in\points_v\cap{(y,y']}}F(p)+F(y'))~~~~[\text{From (\ref{ieq:3}) and (\ref{ieq:4})}]\\
            &= \utilout^{,L}(y,\strat_{-(i-1)}')+\utilout^{,R}(y,\strat_{-(i-1)}')+F(y)\\
            &=\utilout(y,\strat_{-(i-1)}')
        \end{align*}
        Hence we have $\utilout(y',\strat_{-(i-1)})\ge\utilout(y,\strat_{-(i-1)}')$ which contradicts that $\strat$ is a PNE.

    Lastly, for any point $y\in(s_i,x)$, $\utilout(y,\strat') = \utilout^{,L}(y,\strat')\le\utilout(y',\strat)$ where $y'$ is any point between $s_i$ and $s_i$'s immediate left voter's location. Therefore no candidate exhibits a better deviation in $(s_i,x)$ which concludes the proof.
\end{proof}

\subsection{Correctness of Algorithm~\ref{algo:1}}

In this section, we prove the following Claim.
\begin{restatable}{clm}{correctness}
 \label{clm:correctness}
 After $i$ iterations of Algorithm \ref{algo:1}, at least $i$ candidates are on locations in $\bit(i)$, and the resulting configuration is a PNE.
\end{restatable}

The proof of the theorem proceeds through a number of lemmas and technical claims for the two algorithms.

The first lemma crucially shows that for~\Cref{algo:2}, while shifting candidate $i$ to the left, we can ignore $B$-boundaries (single reflections) and just consider $B^2$-boundaries from candidates to the right of $s_i$.

\begin{restatable}{lemma}{lmbbtwo}
\label{lm:BB2}
    Suppose $\strat=\{s_1,\dots,s_m\}$ is a PNE such that $s_1<s_2<\dots<s_m$. Let $i$ be a candidate such that $y\in B(s_{i+1},\points_v)$ and $y\in(s_{i-1},s_{i})$. Then there exists $y'\in[y,s_i]$ such that $y'\in B^2(s_{i+2},\points_v)$.
\end{restatable}

We will need the following claim in the proof of~\Cref{lm:BB2}.

\begin{clm}\label{clm:nosharing}
    Suppose $\strat=\{s_1,s_2,\dots,s_m\}$ is a PNE such that $s_1<s_2<\dots<s_m$. Then no two candidates share votes.
\end{clm}

\begin{proof}
    Assume, for the sake of contradiction, that candidate $i$ is the leftmost candidate who is sharing the votes of $p$ with the candidate $i+1$. Define $x = \min(\{y\in B(s_{i-1},\points_v)\mid y\ge s_i\},s_{i+1})$. By definition of candidate $i$, we conclude that $x>s_i$. Now if candidate $i$ moves to $(x+s_i)/2$, her utility strictly increases, as she gets all the votes from $p$ without losing any votes from the previous configuration. Hence it contradicts that $\strat$ is an equilibrium.
\end{proof}

\begin{proof}[Proof of~\Cref{lm:BB2}]
    Suppose $f_r(s_{i+1},p) = y$ and $\gamma = (s_i-y)$. Note that all the voters at $p$ vote for the candidate $i$. 
    \begin{clm}\label{clm:BB2-1}
        There exists a $q\in\points_v$ such that $f_r(s_{i+2},q)\in[s_{i+1}-\gamma,s_{i+1})$.
    \end{clm}
    \begin{proof}
        From Claim~\ref{clm:nosharing}, we know that there is no $q'\in\points_v$ such that $f_r(s_{i+2},q')=s_{i+1}$. For the sake of contradiction suppose There is no $q\in\points_v$ such that $f_r(s_{i+2},q)\in[s_{i+1}-\gamma,s_{i+1})$. Therefore, candidate $i+1$ can shift more than $\gamma$ to the left and gain all the votes from $p$ without losing any votes. This contradicts that $\strat$ is an equilibrium.
    \end{proof}
    Suppose $q\in\points_v$ such that $f_r(s_{i+2},q)\in[s_{i+1}-\gamma,s_{i+1}]$ (From Claim \ref{clm:BB2-1}). Define $y' =f_r(f_r(s_{i+2},q),\allowbreak p)$. Therefore, $y'\in[y,s_i]$ and $y'\in B^2(s_{i+2},\points_v)$.
\end{proof}

Recall that in Phase 2, for candidate $j$ that we are trying to shift to the left, the nearest position $x' \in \bit(i)$ to the left of $s_j$ is not in $\bit(i-1)$. The next lemma shows that if we cannot move candidate $j$ to the left to position $x' \in \bit(i) \setminus bit(i-1)$ due to a $B^2$-boundary for candidate $k$, then for candidate $k$, the position $x'' \in \bit(i)$ to the immediate left is also not in $\bit(i-1)$.

\begin{restatable}{lemma}{algofourinvariant}
\label{lm:algo4-invariant}
    For any iteration of Algorithm~\ref{algo:2} where Algorithm \ref{algo:2} is called as a subroutine at the $i$th iteration of Algorithm \ref{algo:1}, $s_j\in(x-\epsilon,x)$ and $x-\epsilon\in \bit(i)$.
\end{restatable}

\begin{proof}
    We prove this by induction on the iterations of Algorithm \ref{algo:2}. Note that $\epsilon\in \bit(i)$.
    \paragraph*{Base Case:} Algorithm \ref{algo:2} is invoked as a subroutine in line 15 of Algorithm \ref{algo:1}. At this point, $ s_k \in (x - \epsilon, x) $, where $ (x-\epsilon)\in \bit(i)$.

    \paragraph*{Inductive Case:} Now suppose at some iteration of Algorithm \ref{algo:2}, $s_j\in(x-\epsilon,x)$. If $z=Null$, then Algorithm \ref{algo:2} terminates. Now suppose $Z\neq Null$.
    \begin{clm}\label{clm:algo4-invariant-1}
        If $Z\neq Null$,there exists $x'$ such that either $s_{j+1}\in(x'-\epsilon,x')$ or $s_{j+2}\in(x'-\epsilon,x')$ and $x'\in \bit(i-1)$.
    \end{clm}

    \begin{proof}
        As $Z\neq Null$, there exists $p,q\in\points_v$, such that either $f_r(f_r(s_{j+1},p),q)\in(x-\epsilon,s_j)$ or $f_r(f_r(s_{j+2},p),q)\in(x-\epsilon,s_j)$. Suppose $f_r(f_r(s_{j+1},p),q)\in(x-\epsilon,s_j)$  and $x' = f_r(f_r(x,q),p)$. Note that $s_{j+1}\in(x'-\epsilon,x')$ and $x'\in \bit(i-1)$ since $\points_v\subseteq \bit(0)$ and $x\in \bit(i-1)$. Similarly, it is true if $f_r(f_r(s_{j+2},p),q)\in(x-\epsilon,s_j)$.
    \end{proof}
    From Claim \ref{clm:algo4-invariant-1}, we conclude that line 13 and line 16 of Algorithm \ref{algo:2} correctly find an $x'$ such that either $s_{j+1}\in(x'-\epsilon,x')$ or $s_{j+2}\in(x'-\epsilon,x')$ and $(x'-\epsilon)\in \bit(i)$.
\end{proof}

The last lemma shows that when~\Cref{algo:2} terminates, the strategy profile is an equilibrium.

\begin{restatable}{lemma}{algofourcorrectnessone}
\label{lm:algo4-correctness-1}
If $z=Null$ in Algorithm \ref{algo:2}, either $s_{j-1}=(x-\epsilon)$ or $s_j=(x-\epsilon)$ and the corresponding configuration is an equilibrium.
\end{restatable}

\begin{proof}
    From Lemma \ref{lm:algo4-invariant}, we know that $s_j\in(x-\epsilon,x)$. If $z= Null$, then the set $\{y\in B^2(s_{j+1},\points_v)\cup B^2(s_{j+2},\points_v)~|~s_{j}\ge y\ge(x-\epsilon)\}$ is empty. From Lemma \ref{lm:BB2}, we conclude that the set $\{y\in B(s_{j+1},\points_v)\cup B^2(s_{j+1},\points_v)\cup B^2(s_{j+2},\points_v)~|~s_{j}\ge y\ge(x-\epsilon)\}$ is also empty. If $(x-\epsilon)\le s_{j-1}$, then $s_{j-1}=(x-\epsilon)$; otherwise, $s_j=(x-\epsilon)$. In both cases, from \Cref{lm:3}, we conclude that the resulting configuration is an equilibrium.
\end{proof}

The following claims then establish the correctness of~\Cref{line:shiftboth} in~\Cref{algo:1}, when candidates $k$ and $k+1$ are shifted simultaneously. Suppose $\strat_t$ denotes the set of strategies of $m$ candidates after $t$ iterations of the while loop of Algorithm \ref{algo:1}.

\begin{clm}\label{clm:null}
    Suppose at iteration $t$ of the while loop of Algorithm \ref{algo:1}, $k=Null$ then, $\forall j\in[m]$ $s_j\in \bit(t-1)$.
\end{clm}
\begin{proof}
    For the sake of contradiction suppose there is a candidate $j'\in[m]$ such that $s_{j'}\notin \bit(t-1)$. Therefore, at the iteration $t$ of the while loop of the Algorithm \ref{algo:1}, $s_{j'}\in A = \{j\in[m]~|~s_j\notin \bit(t-1)\}$. Note that $k=min(A)$, which implies $k\neq Null$.
\end{proof}
\begin{clm}\label{clm:subset}
    Suppose at iteration $t$ of the while loop of Algorithm \ref{algo:1}, $k\neq Null$ then, $B(s_{k-1},\points_v)\cup B^2(s_{k-1},\points_v)\cup B^2(s_{k-2},\points_v)\cup \points_v\subseteq \bit(t-1)$.
\end{clm}
\begin{proof}
    Note that by the definition of $k$, all the candidates $k'<k$ are positioned at locations that are in $\bit(t-1)$. Therefore, for all $k'<k$, $s_{k'}\in \bit(t-1)$. Since all the points in $\points_v$ are in $\bit(0)$, $\points_v\subseteq bits(t-1)$. Now from the definitions of $B()$ and $B^2()$, we conclude that for all $k'<k$, $B(s_{k'},\points_v)\subseteq \bit(t-1)$ and $B^2(s_{k'},\points_v)\subseteq \bit(t-1)$ which concludes the proof. 
\end{proof}
\begin{clm}\label{clm:x}
    Suppose $\strat_t$ is a PNE. Then, at the $(t+1)$th iteration of the while loop of the Algorithm \ref{algo:1}, $x\in [s_k,s_{k+2}]$.
\end{clm}
\begin{proof}
    Note that $\util(s_{k+1},\strat_t)>0$ otherwise $\strat_t$ is not a PNE. Therefore there exists $p\in\points_v$ such that $F(p)>0$ and $p\in[s_k,s_{k+2}]$. We know that $\points_v\subseteq \bit(0)$, therefore by the definition of $x$, $x\in[s_k,s_{k+2}]$.
\end{proof}
\begin{clm}\label{clm:e}
   Consider the iteration $t$ of the while loop Algorithm \ref{algo:1}. Then, $\bit(t-1)\cap[x-\epsilon,x)=\emptyset$.
\end{clm}
\begin{proof}
     Note that $x\in \bit(t-1)$. Recall that, $\epsilon=2^{-t}$ and $\epsilon$ is the smallest number that is in $\bit(t)$. Therefore, any points in the interval $[x-\epsilon,x)$ is not in $\bit(t-1)$. 
\end{proof}
\begin{clm}\label{clm:e-small}
    Suppose $\strat_t=\{s_1,...,s_m\}$ is a PNE. Then, at the $(t+1)$th iteration of the while loop of the Algorithm \ref{algo:1}, if $(x-\epsilon)\ge s_{k+1}$, $\strat_t' = \{s_1,...,s_k, x-\frac{\epsilon}{2},s_{k+2},...,s_m\}$ is a PNE.
\end{clm}
\begin{proof}
    If $(x-\epsilon)\ge s_{k+1}$, then $x\in(s_{k+1},s_{k+2}]$ (Claim \ref{clm:x}). We conclude that there is no voter's location in the interval $[s_k,x)$, otherwise it contradicts the definition of $x$. Therefore, $(B(s_{k},\points_v)\cup B^2(s_{k},\points_v) \cup\points_v)\cap [s_{k+1},x)=\emptyset$. Also from Claim \ref{clm:subset} and the definition of $x$, we conclude that $B^2(s_{k-1},\points_v)\cap [s_{k+1},x)=\emptyset$. Therefore, from Lemma \ref{lm:3}, we conclude that, $\strat_t' = \{s_1,...,s_k, x-\frac{\epsilon}{2},s_{k+2},...,s_m\}$ is a PNE.
\end{proof}
\begin{clm}\label{clm:e-exact}
    Suppose $\strat_t=\{s_1,...,s_m\}$ is a PNE. Then, at the $(t+1)$th iteration of the while loop of the Algorithm \ref{algo:1}, if $s_{k+1}>(x-\epsilon)\ge s_k$, $\strat_{t+1} = \{s_1,...,s_{k-1}, x-\epsilon,s_{k+1},...,s_m\}$ is a PNE.
\end{clm}
\begin{proof}
    We conclude the following from the definition of $x$ and Claim \ref{clm:subset}.
    \[
    \left(B(s_{k-1},\points_v)\cup B^2(s_{k-1},\points_v)\cup B^2(s_{k-2},\points_v)\cup \points_v \right)\cap [s_k,x)=\emptyset
    \]
    Note that $[s_k,x-\epsilon)\subseteq[s_k,x)$ and $(x-\epsilon)<s_{k+1}$. Therefore, from Lemma \ref{lm:3}, we conclude that $\strat_{t+1} = \{s_1,...,s_{k-1}, x-\epsilon,s_{k+1},...,s_m\}$ is a PNE.
\end{proof}

Now we prove Claim \ref{clm:correctness} by induction on the iterations of Algorithm \ref{algo:1}.
\begin{proof}
\hphantom{1cm}
\paragraph*{Base Case:} Note that for $i=0$ (iteration 0), Claim \ref{clm:correctness} holds trivially.
\paragraph*{Inductive Case:} Suppose at iteration $i=t+1$, $\strat_t$ is a PNE where at least $t$ candidates are on some locations which are in $\bit(t)$.
\begin{enumerate}
    \item Now if the $t+1$-th iteration of Algorithm \ref{algo:1} terminates on line 5, then by Claim \ref{clm:null}, all the candidates are on locations in $\bit(t)$.
    \item If the $t+1$-th iteration of Algorithm \ref{algo:1} terminates on line 13, then by Claim \ref{clm:e-exact}, $s_k$ is shifted to a location which is in $\bit(t+1)$ and the corresponding configuration is a PNE. Therefore at least $t+1$ candidates are on locations in $\bit(t+1)$.
    \item Suppose Algorithm \ref{algo:1} calls the subroutine of Algorithm \ref{algo:2} at line 15. Note that Algorithm \ref{algo:2} always terminates as either $Z=Null$ or the value of $j$ increases and we know that for $j=m$, $z=Null$. From Lemma \ref{lm:algo4-correctness-1} and \ref{lm:algo4-invariant}, we can conclude that when Algorithm \ref{algo:2} terminates, some candidate is shifted (Or found) to a location which is in $\bit(t+1)$ and the resulting configuration is a PNE. Hence the inductive hypothesis is true for iteration $t+1$, which concludes our proof. \qedhere
\end{enumerate}
\end{proof}

\subsection{Proof of Theorem \ref{thm:appximp}}
\thmappximp*

\begin{proof}

Given $m$ candidates, we construct an instance with $m+3$ voters as shown in Figure~\ref{fig:appximp}. By an argument analogous to Example~\ref{example:1}, we can show that the configuration in Figure~\ref{fig:appximp} is a pure Nash equilibrium.

We now prove that if candidates are restricted to locating only at voter positions, no equilibrium better than a $\frac{1}{m+3}$-equilibrium exists.\footnote{This requires rescaling of votes so that the total mass of votes is 1. Since the instance has $m+3$ votes, we will equivalently show that some candidate can increase her vote share by 1.} 

Now consider any strategy profile where candidates are at voter locations. Observe that there are $m+1$ distinct voter locations. Since there are $m$ candidates, exactly one voter location must be unoccupied. We analyze two cases depending on which location is empty.

\paragraph{Case I:} Suppose one of the extreme voter locations is unoccupied. Then there exists a candidate receiving at most one vote. By deviating to the empty extreme location, that candidate would obtain two votes, increasing her vote share by 1.

\paragraph{Case II:} Suppose an empty location is not an extreme voter location. By construction of the instance, the voter at this location votes for either the immediate left or the immediate right candidate. In either case, the other adjacent candidate can deviate to the mid point of the empty location and her current location, and capture this voter without losing any of her current votes. This again increases her vote share increases by 1.

For instance, suppose the empty location is $b$ (see Figure~\ref{fig:appximp}). Then the voter at $b$ votes for the candidate at $a$. However, the candidate at $d$ can deviate to the midpoint between $c$ and $d$, thereby capturing the vote of the voter at $b$ without losing any of her existing votes.

\begin{figure}[ht]
 \centering
    \begin{tikzpicture}[line cap=round,line join=round,>=triangle 45,x=1cm,y=1cm,scale=0.6]
    \clip(-1,-2) rectangle (22,2);
\draw [line width=1pt] (4,0) -- (19,0);
\draw [line width=1pt,color=black,dashed] (1,0) -- (4,0);

\draw [line width=0.7pt,color=blue,dashed] (11,-0.7) -- (11,0.7);
\draw [line width=0.7pt,color=blue,dashed] (13,-0.7) -- (13,0.7);
\draw [line width=0.7pt,color=blue, dashed] (16,-0.7) -- (16,0.7);

\begin{footnotesize}
\draw [fill=red] (1,0) circle (4pt);
\draw [fill=red] (4,0) circle (4pt);
\draw [fill=red] (5,0) circle (4pt);
\draw [fill=red] (6.5,0) circle (4pt);
\draw [fill=red] (7.5,0) circle (4pt);
\draw [fill=red] (9,0) circle (4pt);
\draw [fill=red] (10,0) circle (4pt);
\draw [fill=red] (12,0) circle (4pt);
\draw [fill=red] (14.5,0) circle (4pt);
\draw [fill=red] (17.5,0) circle (4pt);
\draw [fill=red] (19,0) circle (4pt);

\node[isosceles triangle,
    draw,
    rotate=90,
    fill=black,
    minimum size =0.05cm, scale=0.4] (T1)at (1,1.1){};

\node[isosceles triangle,
    draw,
    rotate=90,
    fill=black,
    minimum size =0.05cm, scale=0.4] (T1)at (4,1.1){};

\node[isosceles triangle,
    draw,
    rotate=90,
    fill=black,
    minimum size =0.05cm, scale=0.4] (T1)at (5,1.1){};

\node[isosceles triangle,
    draw,
    rotate=90,
    fill=black,
    minimum size =0.05cm, scale=0.4] (T1)at (6.5,1.1){};

\node[isosceles triangle,
    draw,
    rotate=90,
    fill=black,
    minimum size =0.05cm, scale=0.4] (T1)at (7.5,1.1){};

\node[isosceles triangle,
    draw,
    rotate=90,
    fill=black,
    minimum size =0.05cm, scale=0.4] (T1)at (9,1.1){};

\node[isosceles triangle,
    draw,
    rotate=90,
    fill=black,
    minimum size =0.05cm, scale=0.4] (T2)at (10,1.1){};
\node[isosceles triangle,
    draw,
    rotate=90,
    fill=black,
    minimum size =0.05cm, scale=0.4] (T3)at (13,1.1){};
\node[isosceles triangle,
    draw,
    rotate=90,
    fill=black,
    minimum size =0.05cm, scale=0.4] (T5)at (17.5,1.1){};

\node[isosceles triangle,
    draw,
    rotate=90,
    fill=black,
    minimum size =0.05cm, scale=0.4] (T1)at (19,1.1){};

\end{footnotesize}
\begin{footnotesize}

\draw[color=blue] (4.5,0.5) node {$2$};
\draw[color=blue] (5.75,0.5) node {$3$};
\draw[color=blue] (7,0.5) node {$2$};
\draw[color=blue] (8.25,0.5) node {$3$};
\draw[color=blue] (9.5,0.5) node {$2$};
\draw[color=blue] (10.5,0.5) node {$2$};
\draw[color=blue] (11.5,0.5) node {$2$};
\draw[color=blue] (12.5,0.5) node {$2$};
\draw[color=blue] (13.75,0.5) node {$3$};
\draw[color=blue] (15.25,0.5) node {$3$};
\draw[color=blue] (16.75,0.5) node {$3$};
\draw[color=blue] (18.25,0.5) node {$3$};

\foreach \x in {4,5,6.5,7.5,9,10,12,14.5,17.5} {
        \draw (\x,-0.8) node[circle, inner sep=0.7pt, draw, color=red, fill=white] {1}; 
    }

\foreach \x in {1,19} {
        \draw (\x,-0.8) node[circle, inner sep=0.7pt, draw, color=red, fill=white] {2}; 
    }

\draw[color=black] (12,-1.8) node {$a$};
\draw[color=black] (14.5,-1.8) node {$b$};
\draw[color=black] (16,-1.8) node {$c$};
\draw[color=black] (17.5,-1.8) node {$d$};

    \end{footnotesize}

\end{tikzpicture}
    \caption{An equilibrium with $m$ candidates and $m+3$ voters. The red circles indicate the voters' locations, with the numbers below each circle representing the number of voters at that location. The blue numbers correspond to the lengths of the respective line segments.}
    \Description{An equilibrium with $m$ candidates and $m+3$ voters. The red circles indicate the voters' locations, with the numbers below each circle representing the number of voters at that location. The blue numbers correspond to the lengths of the respective line segments.}
    \label{fig:appximp}
\end{figure}

\end{proof}


\section{Proofs from Section \ref{sec:bothcontinuous}}

\alphainterval*

\begin{proof}
We first show that every candidate is getting at least $1/(4m) \ge \epsilon$ votes. Since there is a unit total mass of voters, in any strategy profile, there is a voter $i$ that gets at least $1/m$ votes. By moving (arbitrarily close) to the immediate left or right of voter $i$, a voter can hence get at least $1/(2m)$ votes. Since $\strat$ is an $\epsilon$-equilibrium, each voter is hence getting at least $1/(2m) - \epsilon$ votes, which is at least $1/(4m) \ge epsilon$.

The total mass of voters in any $\alpha$-interval is at most $M \times \alpha = \epsilon/8$. Hence if there are 3 candidates in any set of 3 consecutive $\alpha$-intervals, i.e., in the interval $(k\alpha, (k+3)\alpha]$, the candidate in the middle will get at most $3\epsilon/8$ votes, and $\strat$ cannot be an $\epsilon$-equilibrium. This completes the proof.
\end{proof}

\alphadiscretize*

\begin{proof}
    Any candidate moves by at most $\alpha$ relative to $\strat$, and consequently, the mid-point between two candidates also moves by at most $\alpha$. Any candidate receives all the votes between their position and the mid-point to the neighboring candidate on either side; since both these locations change by at most $\alpha$, their total votes received change by at most $2 M \alpha = \epsilon / 2$. Further, for any candidates $i$, $j$, since $s_i$ and $s_j$ change by at most $\alpha$, $|\intu(s_i, s_j) - \intu(s_i',s_j')| \le 2M \alpha$ = $\epsilon/2$. Since $\strat$ is an $\epsilon$-equilibrium, it follows that in $\strat'$, any deviation can improve a candidate's utility by at most $\epsilon + 2 \times \epsilon/2$ $=2\epsilon$.
\end{proof}


\section{Allowing multiple candidates at the same location}
\label{sec:multiplecandidates}

In this section, we present algorithms for deciding existence of an equilibrium (or an approximate equilibrium, as appropriate) when multiple candidates can choose the same location. We show that the algorithms from the main paper, when each candidate must choose a distinct location, can be suitably modified to work for this case as well, though they get significantly more complicated. Many of the proofs in this section are similar to earlier proofs, and hence we only describe the differences from earlier proofs.

For multiple candidates, we need to change the definition of a strategy profile.

\subsection{Redefining notation} 

A strategy profile $\Psi = (\strat, \lambda)$ consists of a vector $\strat = (s_1, \ldots, s_k)$ of $k$ distinct positions in $\points_c$, ordered so that $s_1 < s_2 < \cdots < s_k$, together with a vector $\lambda = (\lambda_1, \ldots, \lambda_k)$ where $\lambda_i>0$ denotes the number of candidates choosing strategy $s_i$. Note that $\sum_{i}\lambda_i=m$. We define $\lambda_0=\lambda_{k+1}=0$. Further, for a strategy profile $\strat$ and $s \in \strat$, we define $\ind(s)$ as the coordinate index of $s$ in the vector $\strat$. If $s \not \in \strat$, we define $\ind(s) = \bot$.

For a strategy $s \in \points_c$, we define $\Psi + s = \Psi' = (\strat', \lambda')$ as
\begin{itemize}
\item If $\ind(s) = j \neq \bot$, then $\strat' = \strat$ and $\lambda' = (\lambda_1, \ldots, \lambda_{j-1}, \lambda_j + 1, \lambda_{j+1}, \ldots, \lambda_k)$.

\item If $\ind(s) = \bot$, then $\strat'$ is the $(k+1)$-tuple obtained by inserting $s$ into $\strat$ and sorting the result in ascending order. The corresponding $\lambda'$ is constructed by inserting $1$ at the position of $s$ in $\strat'$, and keeping the remaining entries of $\lambda$ in order as in $\strat'$.
\end{itemize}
Similarly, for a strategy $s \in \strat$, let $j = \ind(s)$. We define $\Psi - s = (\strat', \lambda')$ as:
\begin{itemize}
\item If $\lambda_j > 1$, then $\strat' = \strat$ and $\lambda' = (\lambda_1, \ldots, \lambda_{j-1}, \lambda_j - 1, \lambda_{j+1}, \ldots, \lambda_k)$.

\item If $\lambda_j = 1$, then $\strat'$ is the $(k-1)$-tuple obtained by removing $s$ from $\strat$, and $\lambda'$ is the corresponding $(k-1)$-tuple obtained by removing $\lambda_j$ from $\lambda$.

\end{itemize}

Given a strategy profile $\Psi=(\strat,\lambda)$ and $i\in[k]$, the utility for a candidate at position $s_i$ is,
\begin{eqnarray*}
\util(i,\Psi) & := \displaystyle \frac{1}{\lambda_i}\cdot\left(F(\mean(s_{i-1},s_i),\mean(s_i,s_{i+1}))- \lambda_{i-1}\cdot\frac{F(\mean(s_{i-1},s),\mean(s_{i-1},s))}{\lambda_{i-1}+\lambda_i} \right. \\
& \qquad \left. - \lambda_{i+1}\cdot\frac{F(\mean(s,s_{i+1}),\mean(s,s_{i+1}))}{\lambda_i+\lambda_{i+1}}\right) \, ,
\end{eqnarray*}
where $\mu(x,y)=(x+y)/2$. With a slight abuse of notation, we write $\util(s, \Psi)$ to denote $\util(\ind(s), \Psi)$.
Further, given a strategy profile $\Psi=(\strat,\lambda)$ and a point $s\in\points_c$, we define $\utilout(s,\Psi)$ as the utility of a \emph{new} candidate at $s$.
\[
\utilout(s,\Psi) = \util(s,\Psi+s)
\]

A strategy profile $\Psi = (\strat,\lambda)$ is a \emph{pure Nash equilibrium} (or simply an equilibrium) if for all $i\in[k]$, and all locations $s \in \points_c$, $\util(i,\Psi) \ge \utilout(s, \Psi-s_i)$. A strategy profile $\Psi=(\strat,\lambda)$ is an \emph{$\epsilon$-equilibrium} if for all $i\in[k]$, and all locations $s\in\points_c$, $\util(i,\Psi) \ge \utilout(s, \Psi-s_i)-\epsilon$. 

\subsection{Discrete candidate positions}\label{apndx:discrete}
\begin{restatable}{prop}{}
\label{prop:equilibrium-characterization-mul}
    Strategy profile $\Psi=(\strat,\lambda)$ where $\strat=(s_1,\ldots,s_k)$ and $\lambda=(\lambda_1,\ldots,\lambda_k)$ is an equilibrium if and only if
    \begin{enumerate}[label=(\roman*)]
        \item for all $i\in[k]$, \[\util(s_{i},\Psi)\ge\max_{s'\in\points_c\cap [s_{i-1},s_{i+1}]} \utilout(s',\Psi-s_i)\]
        \item $\min_{s \in \strat} \util(s,\Psi) \ge \max_{s'\in\points_c} \utilout(s',\Psi)$, i.e., the minimum utility for any candidate is at least the maximum utility that can be obtained by a new candidate at any position.
    \end{enumerate}

\end{restatable}

\begin{proof}
    First, we prove that the conditions in Proposition \ref{prop:equilibrium-characterization-mul} are sufficient. For the sake of contradiction suppose that $(\strat,\lambda)$ is not an equilibrium, therefore, there exists $i\in[k]$  and a location $s'\in\points_c$ such that $\util(i,\Psi)<\utilout(s',\Psi-s_i)$. If $s\in[s_{i-1},s_{i+1}]$, then it clearly violates the first condition of Proposition \ref{prop:equilibrium-characterization-mul}. Now suppose $s\notin [s_{i-1},s_{i+1}]$. Note that the utility of a new candidate on some location $s'$ only depends on the number of candidates already at that location, as well as the number of candidates immediately to the left and right of it. Therefore, $\utilout(s',\Psi-s_i)=\utilout(s',\Psi)$. Then $\util(s_{i},\Psi)< \utilout(s',\Psi)$ which contradicts the second condition of the proposition.
    
    Now we prove that the conditions in Proposition \ref{prop:equilibrium-characterization-mul} are necessary. For the sake of contradiction suppose condition (i) of Proposition \ref{prop:equilibrium-characterization-mul} is not true, i.e., for a candidate at $s_i$ there exists a better local deviation in $[s_{i-1},s_{i+1}]$ which contradicts that $\Psi$ is an equilibrium. Similarly, if condition (ii) fails, then the candidate with minimum utility would have a better deviation within the set $\points_c$, again contradicting that $\Psi$ is an equilibrium. 
\end{proof}
The next proposition shows that the number of possible different utility values is polynomially bounded. Given an instance, $(\points_c,m)$ and an oracle access to $F$, let $\mcV$ be the set of all possible utility values obtainable by a candidate in any strategy profile. Thus $\mcV=\{\util(i,\Psi)\mid \Psi \text{ is a strategy profile, } i\in|\strat|\}$.

\begin{restatable}{prop}{}
\label{prop:different-utilities-mul}
   Given an instance $(\points_c, m)$, the number of possible utility values $|\mcV|$ is $\mathcal{O}(n^4m^3)$, where $n = |\points_c|$.
\end{restatable}
\begin{proof}
    As in the proof of Proposition \ref{prop:different-utilities}, suppose $\points_c = {p_1, \ldots, p_n}$ with $p_1 < \ldots < p_n$. Each candidate's utility can be viewed as a closed interval of voters, with the votes within that interval being equally divided among all candidates located at that position. Since candidates are restricted to positions in $\points_c$, the boundaries of such intervals must lie at the midpoints between pairs of points in $\points_c$. Consequently, the total number of such intervals is $\mathcal{O}(n^4)$.

Given an interval $[p,q]$, the utility of a candidate depends on three quantities: (1) the number of candidates sharing the voters at point $p$, (2) the number of candidates sharing the voters at point $q$, and (3) the number of candidates sharing the aggregate votes from the interval $[p,q]$. Based on these three factors, the number of possible different utilities from an interval is at most $\mathcal{O}(m^3)$. Therefore, $|\mcV|$ is $\mathcal{O}(n^4m^3)$. 
\end{proof}

\subsubsection{Dynamic programming algorithm}
The dynamic programming approach closely follows that of the second section of this paper. The only difference is that the table size is larger to capture deviation by a candidate to the same location as another candidate. Recall that the table was denoted by $T$. Each cell of the table is now denoted by a 13-dimensional vector: \[\eta = (\Delta,s_{i-2},s_{i-1},s_{i},s_{i+1},s_{i+2},\lambda_{i-2},\lambda_{i-1},\lambda_{i},\lambda_{i+1},\lambda_{i+2},\minu,\maxd)\]

 For $l \in \{i-2, i-1, i, i+1, i+2\}$, let $s_l$ denote the candidate locations and $\lambda_l$ the corresponding multiplicities at those locations. Now consider the following invariant. 

\begin{property}\label{property:inv}
We maintain the invariant that $T(\eta)=1$ if there exists positions $s_{i+3}<\ldots<s_k\in\points_c$ and multiplicities $(\lambda_{i+3},\ldots,\lambda_k)$ so that $s_{i+3}>s_{i+2}$ and the following properties hold, for $\strat'=(s_{i-2},s_{i-1},\ldots,s_k)$, $\lambda'=(\lambda_{i-2},\lambda_{i-1},\ldots,\lambda_k)$ and $\Psi' = (\strat',\lambda')$, 

\begin{enumerate}[label = (\roman*)]
    \item for all $j\ge i$,
    \[
    \util(s_{j},\Psi')\ge\max_{s'\in\points_c\cap[s_{j-1},s_{j+1}]}\utilout(s',\Psi'-s_j)
    \]\label{inv:1}
    \item $\minu\ge\maxd$ where,
        \begin{align*}
            \minu &\le \min_{j\ge i} \util(s_{j},\Psi')\\
            \maxd&\ge \max_{s'\in\points_c\cap[s_{i-1},+\infty)}\utilout(s',\Psi')
        \end{align*} \label{inv:2}
\end{enumerate}
\end{property}
Now consider the following cell:
\[
\eta' = (m-\lambda_1-\lambda_2-\lambda_3,-\infty,-\infty,s_{1},s_{2},s_{3},0,0,\lambda_{1},\lambda_{2},\lambda_{3},\minu,\maxd)
\]

\begin{lemma}\label{lm:mul-dp-res}
    Suppose that for some values $s_1, s_{2}, s_{3}, \lambda_1, \lambda_{2}, \lambda_{3}, \minu, \maxd$, we have $T(\eta') = 1$. Then an equilibrium exists. Conversely, if an equilibrium exists, there exist values $s_1, s_{2}, s_{3}$, $\lambda_1, \lambda_{2}, \lambda_{3}$, $\minu, \maxd$ such that $T(\eta') = 1$.
\end{lemma}
\begin{proof}
    Suppose $T(\eta') = 1$. Then, by Property \ref{property:inv}, it follows that we can place $m$ candidates at the locations $\strat = (s_1, \ldots, s_k)$ with corresponding frequencies $\lambda = (\lambda_1, \ldots, \lambda_k)$ and $\Psi=(\strat,\lambda)$ such that
    \begin{enumerate}
        \item for all $j\ge 1$,
    \[
    \util(s_{j},\Psi)\ge\max_{s'\in\points_c\cap[s_{j-1},s_{j+1}]}\utilout(s',\Psi-s_j)
    \]
    \item $\minu\ge\maxd$ where,
        \begin{align*}
            \minu &\le \min_{j\ge 1} \util(s_{j},\Psi)\\
            \maxd&\ge \max_{s'\in\points_c\cap[s_0,+\infty)}\utilout(s',\Psi)
        \end{align*}
    \end{enumerate}
    Note that $\points_c\cap[s_0,+\infty)=\points_c$. From Proposition \ref{prop:equilibrium-characterization-mul}, we conclude that $\Psi$ is an equilibrium. Now suppose, $\Psi=(\strat,\lambda)$ is an equilibrium. Take $\minu=\min_{i\in[k]}\util(s_{i},\Psi)$ and $\maxd=\max_{s'\in\points_c}\utilout(s',\Psi)$. Then the cell $\eta'$ satisfies the Property \ref{property:inv}. Therefore, $T(\eta')=1$.
\end{proof}

We now present the dynamic programming algorithm. It iterates only over table entries where the candidate locations are drawn from the set $\points_c \cup \{-\infty, +\infty\}$. Any configuration in which a multiplicity is negative or where the multiplicity at any location—or the total sum of multiplicities—exceeds $m$, is assigned a table value of 0. For each cell $\eta$, the algorithm computes its value based on the values of other cells where the next location $s_{i+3}$ lies strictly to the right of $s_{i+2}$ therefore, the algorithm iterates from right to left in the time of filling the table entries. Note that for any cell $\eta$, we have $s_i \in \points_c$, $s_{i-2},s_{i-1}\in\points_c\cup\{-\infty\}$, and $s_{i+1},s_{i+2}\in\points_c\cup\{+\infty\}$.
\paragraph*{Base Case:} The base case for our table filling would be when for any cell $\eta$, $\Delta = 0$. We set $T(\eta)=1$ if it satisfies all the conditions in Property \ref{property:inv}.
Note that in this case, there is no candidate after $s_{i+2}$. Therefore, we can easily check the two conditions in Property \ref{property:inv} in polynomial time.
\paragraph*{Inductive case:} For any cell $\eta$ with $\Delta >0$, we set $T(\eta)=1$ if the following hold:
\begin{enumerate}[label=(\roman*)]
    \item Let $\strat^i =(s_{i-2},s_{i-1},s_{i},s_{i+1},s_{i+2})$, $\lambda^i=(\lambda_{i-2},\lambda_{i-1},\lambda_i,\lambda_{i+1},\lambda_{i+2})$, and $\Psi^i = (\strat^i,\lambda^i)$. Then,
    \[
    \util(s_{i},\Psi^i)\ge \max_{s'\in\points_c\cap[s_{i-1},s_{i+1}]}\utilout(s',\Psi^i-s_i)
    \]
    \item $\minu\le\util(s_{i},\Psi^i)$ and $\maxd\ge\max_{s'\in\points_c\cap[s_{i-1},s_{i+2})}\utilout(s',\Psi^i)$
    \item There exists $s_{i+3}>s_{i+2}$, $\Delta\ge\lambda_{i+3}>0$, $\minu'$, $\maxd'$ with $\minu'\ge$ $\minu\ge\maxd\ge\maxd'$ such that
    \[
    T(\Delta-\lambda_{i+3},s_{i-1},s_{i},s_{i+1},s_{i+2},s_{i+3},\lambda_{i-1},\lambda_i,\lambda_{i+1},\lambda_{i+2},\lambda_{i+3},\minu',\maxd')=1
    \]
\end{enumerate}

\begin{lemma}
    For a cell $\eta$, the algorithm sets $T(\eta)$ to 1 iff all the conditions hold in Property \ref{property:inv} for the cell $\eta$.
\end{lemma}\label{lm:dp-mul}
\begin{proof}
    The proof closely follows the structure of the proof of Lemma \ref{lem:discrete-table-filling}; therefore, we provide only a brief outline. We begin by proving the `if' direction of Lemma \ref{lm:dp-mul} using induction.

    \paragraph*{Base case:} Suppose for $\eta$, $\Delta = 0$. Then, from the definition of the Base of the algorithm, $T(\eta) = 1$.
    \paragraph*{Inductive case:} As Property \ref{property:inv} holds for $\eta$, there exists $\strat'=(s_{i-2},s_{i-1},\ldots,s_k)$, $\lambda=(\lambda_{i-2},$ $\lambda_{i-1},\ldots,\lambda_k)$, $\minu$, and $\maxd$ such that
    \begin{enumerate}[label = (\roman*)]
    \item for all $j\ge i$,
    \[
    \util(s_{j},\Psi')\ge\max_{s'\in\points_c\cap[s_{j-1},s_{j+1}]}\utilout(s',\Psi'-s_j)
    \]\label{inv:1-mul}
    \item $\minu\ge\maxd$ where,
        \begin{align*}
            \minu &\le \min_{j\ge i} \util(s_{j},\Psi')\\
            \maxd&\ge \max_{s'\in\points_c\cap[s_{i-1},+\infty)}\utilout(s',\Psi')
        \end{align*} \label{inv:2-mul}
\end{enumerate}
where $\Psi'=(\strat',\lambda')$. Suppose that for the cell $\eta$, we have $\Delta = t > 0$, and assume that the algorithm has correctly assigned values for all cells with $\Delta' < t$. Then, all three conditions of the inductive step of the algorithm are satisfied for $\Psi'$, $\minu$, and $\maxd$, which implies that $T(\eta) = 1$.

    We now prove the `only if' direction. In the base case, $T(\eta) = 1$ with $\Delta = 0$, which trivially satisfies all the conditions of Property \ref{property:inv}. Now, assume the `only if' direction holds for all cells with $\Delta' \le t - 1$. Let $\Delta = t$ for the cell $\eta$. If $T(\eta) = 1$, then there exist values $s_{i+3}$, $\lambda_{i+3} > 0$, $\minu'$, and $\maxd'$ such that the corresponding table entry $(\Delta-\lambda_{i+3},s_{i-1},s_{i},s_{i+1},s_{i+2},s_{i+3},\lambda_{i-1},\lambda_i,\lambda_{i+1},\lambda_{i+2}$, $\lambda_{i+3},\minu',\maxd')$ is 1, and for that entry, $\Delta-\lambda_{i+3} < t$. By the inductive hypothesis, all conditions in Property \ref{property:inv} hold for that cell, which in turn implies that they also hold for $\eta$ as $\minu'\ge\minu\ge\maxd\ge\maxd'$.
\end{proof}

\subsection{Continuous candidate and discrete voter locations}\label{apndx:con-dis-mul}

In this section, we assume that the set of voters' locations $\points_v\subseteq \{0,1,\ldots, R\}$ and $|\points_v|=n$. The candidates can choose any location from $[0,R]$ as a strategy (thus $\points_c = [0,R])$. Here we modify~\Cref{algo:1} and~\Cref{algo:2} so that given an equilibrium $\strat$ with possibly multiple candidates at the same point, the algorithms return an equilibrium $\strat'$ where each strategy $s\in \strat'$ is in $\bit(m)$, where $m$ is the number of candidates. Given the existence of such an equilibrium, we can then use the modified dynamic program to decide existence of an equilibrium in time exponential in $m$ and polynomial in $R$ and $n$. 

As in the previous sections, we first need to redefine some of our notation. Let $\Psi=(\strat,\lambda)$ be a strategy profile with $k = |\strat|$ the number of distinct locations. For $i\in[k]$, we define $v^L(i)$ and $v^R(i)$ as the location of voters to the immediate left and immediate right of $s_i$, i.e., $v^L(i) = \max \{p \in \points_v ~ : ~ p < s_i\}$ and $v^R(i) = \min \{p \in \points_v ~ : ~ p > s_i\}$. As with other notation, for $s \in \strat$, we use $v^L(s)$ and $v^R(s)$ to denote $v^L(\ind(s))$ and $v^R(\ind(s))$ respectively. 

We define the left and right utilities for a strategy in $\strat$ as

\begin{align*}
    \utileft(i,\Psi)&= F(\mu(s_{i-1},s_i),v^L(i)) - \lambda_{i-1}\cdot\frac{F(\mu(s_{i-1},s_i),\mu(s_{i-1},s_i))}{\lambda_{i-1}+\lambda_i}\\
    \utilright(i,\Psi)&= F(v^R(i),\mu(s_{i},s_{i+1})) - \lambda_{i+1}\cdot\frac{F(\mu(s_{i},s_{i+1}),\mu(s_{i},s_{i+1}))}{\lambda_{i+1}+\lambda_i}
\end{align*}

Observe that,
\[
\util(i,\Psi) = \frac{1}{\lambda_i}\cdot\left(F(s_i,s_i)+\utileft(s_i,\strat,\lambda)+\utilright(s_i,\strat,\lambda)\right).
\]

We begin by noting that at an equilibrium, if there is no voter at a location, then there can be no more than two candidates at this location. Further if there are two candidates at this location, then the left and right utilities at the location must be equal.

\begin{lemma}\label{lm:tower-equality}
   Suppose $\Psi=(\strat,\lambda)$ is a PNE. Let $i\in[k]$ be such that $\lambda_i \ge 2$ and $s_i\notin\points_v$. Then $\lambda_i = 2$, and $\utileft(i,\Psi)=\utilright(i,\Psi)$.
\end{lemma}
\begin{proof}
    Since $s_i \not \in \points_v$, $F(s_i, s_i) = 0$. Then for the first inequality, note that at $s_i$, the utility for any candidate is $\util(i, \Psi) = \frac{1}{\lambda_i} \left( \utileft(i, \Psi) + \utilright(i, \Psi) \right)$. If $\lambda_i \ge 3$, then $\util(i, \Psi) < \max \{\utileft(i, \Psi),$ $ \utilright(i, \Psi)\}$, and a candidate at $s_i$ can increase her utility by shifting slightly to the left (or right).
    
    For the second equality, if $\utileft(i, \Psi) > \utilright(i, \Psi)$, then
\[
\util(i, \Psi) = \frac{1}{\lambda_i} \left( \utileft(i, \Psi) + \utilright(i, \Psi) \right) < \utileft(i, \Psi),
\]

\indent and a candidate at $s_i$ can deviate slightly to the left and increase her utility to $\utileft(i, \Psi)$. A similar contradiction holds if $\utileft(i, \Psi) > \utilright(i, \Psi)$.
\end{proof}

We now modify Lemmas \ref{lm:1} and \ref{lm:3} to allow for multiple candidates at the same location. The proofs are very similar to the original proofs, and hence we sketch out the differences from the original proofs. We use the definition of reflections of candidates $B(s_i, \points_v)$ and reflections of reflections $B^2(s_i, \points_v)$ as in the original lemmas.

\begin{restatable}{lemma}{lmonemul}
\label{lm:mul-1}
    Let $\Psi=(\strat,\lambda)$ be an equilibrium with $k = |\strat|$. Fix an $i\in[k]$, and let $A= B(s_{i-1},\points_v)\cup \points_v$ be the set of voter positions and the reflections of $s_{i-1}$.  If $s_i \not \in A$, let $x$ be such that $s_i<x<s_{i+1}$ and $A\cap[s_i,x]=\emptyset$, and let $\strat'=(s_1,\dots,s_{i-1},x,s_{i+1},\dots,s_k)$. Then each candidate has the same utility in $\Psi=(\strat,\lambda)$ and $\Psi'=(\strat',\lambda)$. 
\end{restatable}
\begin{proof}[Proof sketch]
   The proof closely follows the argument in Lemma \ref{lm:1}. The difference is that while Lemma \ref{lm:1} considers shifting a single candidate from $s_i$ to $x$, here we move all $\lambda_i$ candidates from $s_i$ to $x$. Note that in this process, the utilities of candidates positioned at $s_1, \ldots, s_{i-2}, s_{i+2}, \ldots, s_k$ remain unchanged from their utilities in $\Psi$.

Since the set $A$ is independent of the multiplicity vector $\lambda$, following the argument in~\Cref{lm:1}, the utilities of the candidates at $s_{i-1}$ remain unchanged. Further, by Lemma \ref{lm:tower-equality}, we know that if any candidate at $s_i$ deviates to $x$, her utility does not decrease.\footnote{Note that $x\in(s_i,v^R(i))$ where as defined $v^R(i)$ is the location of the voter to the immediate right of $s_i$.} Now, observe that the utility of the candidates at $s_{i+1}$ can only decrease if the candidates at $s_i$ are shifted right. However, if it does, then a candidate at $s_i$ would strictly benefit by deviating to $x$, which contradicts the assumption that $\Psi$ is a pure Nash equilibrium. Therefore, all candidates have the same utility in both $\Psi$ and $\Psi'$.
\end{proof}

\begin{restatable}{lemma}{shiftingmul}
\label{lm:mul-3}
    Let $\Psi=(\strat,\lambda)$ be an equilibrium with $k = |\strat|$. Let $i\in[k]$ be such that $s_i\notin A$ where $A=B(s_{i-1},\points_v)\cup B^2(s_{i-1},\points_v)\cup B^2(s_{i-2},\points_v)\cup \points_v$. Let $x$ be a location such that $s_{i}<x<s_{i+1}$ and $A\cap [s_i,x]=\emptyset$. Then, $\Psi'=(\strat',\lambda)$ is a PNE where $\strat'=(s_1,...,s_{i-1},x,s_{i+1},...,s_k)$.

    Similarly, Let $i\in[|\strat|]$ such that $s_i\notin A$ where $A=B(s_{i+1},\points_v)\cup B^2(s_{i+1},\points_v)\cup B^2(s_{i+2},\points_v)\cup \points_v$. Let $x$ be a location such that $s_{i-1}<x<s_i$ and $A\cap [x,s_i]=\emptyset$. Then, $\Psi'=(\strat',\lambda)$ is a PNE where $\strat'=(s_1,...,s_{i-1},x,s_{i+1},...,s_k)$ is a PNE.
\end{restatable}
\begin{proof}[Proof sketch]
    The proof closely follows the argument in Lemma~\ref{lm:3}. By Lemma~\ref{lm:mul-1}, all candidates have identical utilities in both $\Psi$ and $\Psi'$. Note that any candidate located at a position other than $s_{i-1}$ can have a better deviation in $\Psi'$ only within the interval $(s_{i-1}, x)$. However, following the proof for Case I in Lemma~\ref{lm:3}, if such a deviation exists in $\Psi'$, there also exists a better deviation for that candidate in $\Psi$, which contradicts the assumption that $\Psi$ is an equilibrium. 

Now, suppose for contradiction that a candidate at $s_{i-1}$ has a profitable deviation in $\Psi'$. If $\lambda_{i-1} = 1$, then the situation is analogous to Case II in Lemma~\ref{lm:3}. Consider now the case where $\lambda_{i-1} \ge 2$. The candidate at $s_{i=1}$ can only have a profitable deviation to some point $y \in (s_{i-1}, x)$. As in Case I of Lemma~\ref{lm:3}, define the set
\[
Z = \{ q \in B(s_i, \points_v) \mid y \leq q \leq y + \gamma \},
\]
where $\gamma = x - s_i$. From Claim~\ref{clm:2}, we know that $Z$ is non-empty. Let
\[
b = \min\{ q \in B(s_{i-1}, \points_v) \mid q \geq y \}, \quad b' = \min(s_i, b), \quad \text{and} \quad z_{\max} = \max(Z),
\]
and define $y' = ((z_{\max} + b')/2)$. As shown in Case I of Lemma~\ref{lm:3}, $y'$ constitutes a profitable deviation for a candidate at $s_{i-1}$ in $\Psi$, contradicting the assumption that $\Psi$ is an equilibrium.
\end{proof}

We now show that~\Cref{algo:1} in fact works with multiple candidates at a location. The key difference is that in the initial strategy profile $\strat$ if there are multiple candidates at a location, these are shifted together, and are kept together throughout the algorithm. The algorithm as stated works, the only difference being that $s_i$ is interpreted as the $i$th location in the strategy profile, rather than the $i$th candidate. We then obtain the following claim for the algorithm, similar to~\Cref{clm:correctness}.

\begin{clm}\label{clm:correctness-mul}
    After $i$ iterations of Algorithm \ref{algo:1} on input $(|\strat|,\strat,\points_v)$, at least $i$ candidates are on locations in $\bit(i)$ and the resulting configuration is a PNE.
\end{clm}
\begin{proof}[Proof sketch]
    Let $k = |\strat|$. The proof builds heavily on Lemma~\ref{lm:mul-3}. Algorithm~\ref{algo:1} performs local shifts of candidates in each iteration, which may change the locations of the candidates but preserves their relative order. Suppose that after $t$ iterations of Algorithm~\ref{algo:1}, at least $t$ candidates occupy locations in the set $\bit(t)$, and the resulting location profile is denoted by $\strat^t = (s^t_1, \ldots, s^t_k)$ is a PNE. The algorithm then identifies the smallest index $j \in [k]$ such that $s^t_j \notin \bit(t)$. There are two cases:

\noindent \textbf{Case I:} Suppose there exists a point $x$ such that $s^t_j < x < s^t_{j+1}$, with $A \cap [s^t_j, x] = \emptyset$ and $x \in \bit(t+1)$, where
\[
A = B(s^t_{j-1}, \points_v) \cup B^2(s^t_{j-1}, \points_v) \cup B^2(s^t_{j-2}, \points_v) \cup \points_v.
\]
Then Algorithm~\ref{algo:1} shifts all candidates from $s^t_j$ to $x$. By Lemma~\ref{lm:mul-3}, the resulting strategy profile is a PNE, and now at least $t+1$ candidates are located at positions in $\bit(t+1)$.

\noindent \textbf{Case II:} If no such $x$ satisfying the above conditions exists, the algorithm identifies an index $l$ such that $k\ge l > j$ and $s^t_l \notin \bit(t)$ and there exists $y$ with $s^t_{l-1} < y < s^t_l$, $A' \cap [y, s^t_l] = \emptyset$, and $y \in \bit(t+1)$, where
\[
A' = B(s^t_{l+1}, \points_v) \cup B^2(s^t_{l+1}, \points_v) \cup B^2(s^t_{l+2}, \points_v) \cup \points_v.
\]
In this case, Algorithm~\ref{algo:1} shifts all candidates from $s^t_l$ to $y$, and by Lemma~\ref{lm:mul-3}, the resulting configuration is a PNE. \end{proof}

\subsection{Reducing the strategy space}
Recall the definition of the set $\chi_i$.
\[
\chi_i
=
\left\{
x \pm \frac{k}{2^m}
\,\middle|\,
x \in \{0,\dots,R\},\;
k \in \{0,\dots,2i\}
\right\}\cap [0,R]
\]
and $\chi = \chi_m$. We state a theorem similar to Theorem \ref{thm:dense}.
\begin{restatable}{theorem}{thmdensemul}
\label{thm:dense-mul}
If there is an equilibrium for instance $\mcI=(\points_v,m,F,R)$, then there is an equilibrium $\Psi = (\strat,\lambda)$ where each strategy $s \in \strat$ is in $\chi$.
\end{restatable}

\begin{proof}[proof sketch]
By \Cref{clm:correctness-mul}, if an equilibrium exists for the instance 
$\mcI = (\points_v, m, F, R)$, then there exists an equilibrium profile 
$\strat$ in which every strategy $s \in \strat$ lies in $\bit(m)$. 
We therefore begin with such an equilibrium $\Psi = (\strat,\lambda)$, where all candidates are located at points in $\bit(m)$, and use it as input to Algorithm~\ref{algo:chishifting}; that is, we run the algorithm on $(m,\strat,\points_v)$. 

We now argue that Claim~\ref{clm:chishifting} continues to hold for this input. 
In the setting without multiple candidates at the same location, the proof of Claim~\ref{clm:chishifting} relies on two structural results, namely~\Cref{prop:polym} and Lemma~\ref{lm:3}. 
Observe that~\Cref{prop:polym} holds for the set $\chi$, and Lemma~\ref{lm:mul-3} is the natural analogue of Lemma~\ref{lm:3} in the multi-candidate setting. The latter ensures that, in order to shift candidate(s) from a location to the right while preserving equilibrium, it suffices to verify that no voter location, reflection, or double reflection generated by the two distinct candidate locations to the left is crossed. Hence, after Algorithm~\ref{algo:chishifting} terminates on input $(m,\strat,\points_v)$, all candidates lie at points of the form $x \pm \frac{k}{2^m}$, where $x \in \{0,\ldots,R\}$ and $k \in \{0,\ldots,2m\}$.
\end{proof}

Now we can use the modified dynamic programming from \Cref{sec:continuous-candidates} on the set $\hat{\points_c} = \chi$ to decide the existence of an exact equilibrium.

\subsection{Continuous candidate and voter locations}\label{apndx:con-con-mul}
In this section, we prove a result similar to Theorem~\ref{thm:main-cvcc}. Recall that $M$ is the upper bound on the voter distribution $f(\cdot)$.

\begin{theorem}
    In the model with continuous voter and candidate locations and allowing multiple candidates at the same location, if there exists an $\epsilon$-eqilibrium for $\epsilon>0$, then a $4\epsilon$-equilibrium can be computer in time poly($m,R,M/\epsilon$).
\end{theorem}\label{thm:main-cvcc-mul}

We will use the same approach as in~\Cref{sec:bothcontinuous}. Recall the definition of an $\alpha$-interval from~\Cref{sec:bothcontinuous}, where $\alpha = \epsilon/(4M)$.

\begin{clm}
    In any $\epsilon$-equilibrium $\Psi = (\strat,\lambda)$, there are at most two distinct candidate locations in any $\alpha$-interval. Moreover, if an $\alpha$-interval contains two such locations, then both the preceding and succeeding $\alpha$-intervals must be empty.
    \label{clm:alpha-interval-mul}
\end{clm}
Now, if there is only one candidate location in an $\alpha$-interval, we shift all the candidates to the next $\alpha$-position. If there are two candidate locations, then we shift all the candidates from left candidate location to the left $\alpha$-position and shift all the candidates from right candidate location to the right $\alpha$-position. Suppose the resulting tuple of candidate locations is $\strat'$. Note that the multiplicity vector $\lambda$ stays same.
\begin{clm}
    $(\strat',\lambda)$ is a $2\epsilon$-equilibrium.
\end{clm}\label{clm:alpha-discretize-mul}

The proofs of Claim \ref{clm:alpha-interval-mul} and Claim \ref{clm:alpha-discretize-mul} follow the same arguments as those of Claim \ref{clm:alpha-interval} and Claim \ref{clm:alpha-discretize}, respectively. Now we describe the necessary and sufficient conditions for the existence of an $\epsilon$-equilibrium in the case of allowing multiple candidates on the same location.

\begin{prop}
\label{prop:epsilon-equilibrium-characterization-mul}
    Strategy profile $\Psi=(\strat,\lambda)$ where $\strat=(s_1,\ldots,s_k)$ and $\lambda=(\lambda_1,\ldots,\lambda_k)$ is an $\epsilon$-equilibrium if and only if
    \begin{enumerate}[label=(\roman*)]
        \item for all $i\in[k]$, \[\util(s_{i},\Psi)\ge\max_{s'\in\points_c\cap [s_{i-1},s_{i+1}]} \utilout(s',\Psi-s_i)-\epsilon\]
        \item The minimum utility among all candidates, $\minu\ge\maxd-\epsilon$, where $\maxd =\max_{s'\in\points_c}$ $\utilout(s',\Psi)$ and $\minu = \min\{\util(s,\Psi)\mid s\in\strat\}$.
    \end{enumerate}

\end{prop}
Following~\Cref{prop:epsilon-equilibrium-characterization-mul}, we give the algorithm to find $4\epsilon$-equilibrium given that there exists an $\epsilon$-equilibrium. Recall that we define the vector
\[\eta = (\Delta,s_{i-2},s_{i-1},s_{i},s_{i+1},s_{i+2},\lambda_{i-2},\lambda_{i-1},\lambda_{i},\lambda_{i+1},\lambda_{i+2},\minu,\maxd)\]
For any cell $\eta$ with $\Delta >0$, we set $T(\eta)=1$ if the following hold:
\begin{enumerate}[label=(\roman*)]
    \item Let $\strat^i =(s_{i-2},s_{i-1},s_{i},s_{i+1},s_{i+2})$, $\lambda^i=(\lambda_{i-2},\lambda_{i-1},\lambda_i,\lambda_{i+1},\lambda_{i+2})$, and $\Psi^i = (\strat^i,\lambda^i)$. Then,
    \[
    \util(s_{i},\Psi)\ge \max_{s'\in\points_c\cap[s_{i-1},s_{i+1}]}\utilout(s',\Psi^i-s_i)-\epsilon
    \]
    \item $\minu\le\util(s_{i},\Psi^i)$ and $\maxd\ge\max_{s'\in\points_c\cap[s_{i-1},s_{i+2})}\utilout(s',\Psi^i)$
    \item There exists $s_{i+3}>s_{i+2}$, $\Delta\ge\lambda_{i+3}>0$, $\minu'$, $\maxd'$ with $\minu'\ge$ $\minu\ge\maxd-\epsilon\ge\maxd'-\epsilon$ such that
    \[
    T(\Delta-\lambda_{i+3},s_{i-1},s_{i},s_{i+1},s_{i+2},s_{i+3},\lambda_{i-1},\lambda_i,\lambda_{i+1},\lambda_{i+2},\lambda_{i+3},\minu',\maxd')=1 \, .
    \]
\end{enumerate}

\end{document}